\documentclass[letterpaper,11pt]{article}

\usepackage[T1]{fontenc}
\usepackage{mathpazo,bbm}
\DeclareSymbolFont{calletters}{OMS}{cmsy}{m}{n}
\DeclareSymbolFontAlphabet{\mathcal}{calletters}
\usepackage{microtype}
\usepackage{fullpage}
\usepackage{amsmath}
\usepackage{amsthm}
\usepackage{amsfonts}
\usepackage{amssymb}
\usepackage{mathtools}
\usepackage[usenames,dvipsnames]{xcolor}
\usepackage{hyperref}
\usepackage[noabbrev,capitalize]{cleveref}
\usepackage{svg}
\usepackage{verbatim}
\usepackage{lipsum}
\usepackage{braket}
\usepackage{enumitem}
\usepackage{qcircuit}
\usepackage{array}
\usepackage{thm-restate}
\usepackage{multirow}
\usepackage{booktabs}
\usepackage[linesnumbered,ruled,vlined]{algorithm2e}
\SetKwInput{KwInput}{Input}                
\SetKwInput{KwOutput}{Output}
\SetKwFor{RepTimes}{repeat}{times}{end}
\Crefname{algocf}{Algorithm}{Algorithms}
\hypersetup{
  colorlinks = true,
  linkcolor = {MidnightBlue},
  citecolor = {MidnightBlue},
  urlcolor  = {MidnightBlue}
}
\usepackage{listings}
\usepackage{pgfplots}
\usepgfplotslibrary{fillbetween}

\newtheorem{theorem}{Theorem}[section]
\newtheorem{lemma}[theorem]{Lemma}

\newtheorem{corollary}[theorem]{Corollary}

\newtheorem{proposition}[theorem]{Proposition}
\theoremstyle{definition}
\newtheorem{definition}[theorem]{Definition}
\theoremstyle{remark}

\newtheorem{claim}[theorem]{Claim}

\newcommand{\E}{\mathbb{E}}

\newcommand{\F}{\mathbb{F}}

\newcommand{\calB}{\mathcal{B}}
\newcommand{\calC}{\mathcal{C}}
\newcommand{\calD}{\mathcal{D}}
\newcommand{\calE}{\mathcal{E}}
\newcommand{\calF}{\mathcal{F}}
\newcommand{\calL}{\mathcal{L}}

\newcommand{\calO}{\mathcal{O}}
\newcommand{\calR}{\mathcal{R}}
\newcommand{\calS}{\mathcal{S}}
\newcommand{\calY}{\mathcal{Y}}

\newcommand{\domain}{D}
\newcommand{\range}{R}

\DeclarePairedDelimiterX{\norm}[1]{\lVert}{\rVert}{#1}

\newcommand{\floor}[1]{\left\lfloor #1 \right\rfloor}
\newcommand{\ceil}[1]{\left\lceil #1 \right\rceil}

\newcommand\ignore[1]{\section{Preliminaries}
\label{sec:prelims}
We define the binary entropy function $H_2:[0,1]\rightarrow \mathbb{R}$, by $H_2(p)=- p\log_2 p - (1-p) \log_2 (1-p)$.

\begin{proposition}[Shannon]
\label{prop:shannon}
The number of subsets of $[k]$ of size at most $\alpha k$ is at
most $2^{H_2(\alpha)\, k}$.
\end{proposition}

\begin{definition}
An $m \times n$ matrix is \emph{$(g,h,c)$-rigid} iff every $k \times w$ submatrix where $k \leq g$ and $w \geq n-h$ has rank at least $ck$. We call $(g,h,1)$-rigid matrices $(g,h)$-rigid.
\end{definition}

Matrix rigidity is a robust notion of rank and is an important property for proving time-space and cumulative complexity lower bounds for linear algebra. Fortunately, Yesha gives an explicit example of such a matrix and Abrahamson proved that there are many rigid square matrices.

\begin{proposition}[Lemma 3.2 in \cite{Yes84}]\label{prop:DFT-rigid}
The $n\times n$ Discrete Fourier Transform (DFT) matrix is $(n/4,n/4, 1/2)$ rigid.
\end{proposition}

\begin{proposition}[Lemma 4.3 in \cite{Abr91}]\label{prop:rigid-matrices}
There is a constant $\gamma \in (0, \frac{1}{2})$ such that at least a $1-d^{-1}(2/3)^{\gamma n}$ fraction of the matrices over $\domain^{n \times n}$ with $|\domain|=d$ are $(\gamma n,\gamma n)$-rigid.
\end{proposition}

\subsection{Time space tradeoffs for multi-output functions}\label{subsec:prelim-ts-tradeoffs}

\paragraph{Unitary quantum circuits with oracle states}
Throughout this paper, we consider quantum circuits that seek to compute target functions $f: \domain^n \to \range^{m}$ (or functions
$f:\domain^n\rightarrow \mathcal{P}(\range)$ where the requirement is to
output at least $m$ elements of
$f(x)$ if they exist). Let $d=|\domain|$ and assume the existence of some canonical bijective map $\nu: \domain \to \set{0,\ldots,d-1}$ that gives us an ordering on the elements of $\domain$.
A $T$-query quantum circuit $\calC$ is specified using input independent unitaries $U_0, \ldots, U_T$.
These unitaries define a sequence of quantum states $\ket{\psi_1}_\calC, \ldots \ket{\psi_T}_\calC$ that an algorithm enters during its execution. When it is ambiguous, we use the subscript $\calC$ to denote the partial trace of $\ket{\psi_t}$ that keeps only the qubits involved in the state of the query algorithm.
Note that even though $\ket{\psi_t}$ is always a pure state, $\ket{\psi_t}_\calC$ is often a mixed state.
We can think of each of these states $\ket{\psi_t}_\calC$ as a linear combination of basis vectors $\ket{i,p,w}$ where $i$ represents an index to query, $p$ represents a phase for the query, and $w$ contains all the remaining qubits of the state.

\begin{sloppypar}
Similar to \cite{Amb02, Zha19,HM21}, we define a general oracle operator $\calO$ that interacts with an input register that starts in a state $\ket{\psi_0}_\calO$.
When it is ambiguous, we use the subscript $\calO$ to denote the partial trace of $\ket{\psi_t}$ that keeps only the qubits involved in the state of the  oracle containing the input.
Given a distribution $\calD$ over $D^n$, we can make $\ket{\psi_0}_\calO = \sum_{X \in D^n} \sqrt{\Pr_{X' \sim \calD}[X' = X]} \ket{X}$ to represent an input sampled from $\calD$.
We define our oracle operator $\calO$ as
\begin{equation*}\calO \ket{i,p,w}\ket{X} = \omega_d^{x_i p} \ket{i,p,w}\ket{X}.\end{equation*}
Thus the joint state of the input and quantum circuit at the end of the computation is given by
\begin{math}\ket{\psi_T} = U_T \calO \ldots \calO U_0 \ket{\psi_0}\end{math}
where $\ket{\psi_0}=\ket{0}_\calC \otimes \ket{\psi_0}_\calO$.
\end{sloppypar}

The output of the quantum circuit is determined by measuring the work register of $\ket{\psi_T}_\calC$ in the standard basis and applying some input-independent post-processing function $q$ to interpret the result as an output $\tau \in R^J$ where $J \subseteq [m]$.
The correctness of these output values is then determined by measuring the input registers in the standard basis to obtain the input $X$ and evaluating whether $\tau$ is consistent with $f(X)$,
which we denote by writing $\tau \| f(X)$.
In general we can define the projector $\Pi_{k}$ where:
\begin{equation}\label{eq:output-k}
\Pi_{k} = \mkern-40mu \sum_{\substack{i,p,w,x_1, \ldots, x_n \\ \text{s.t. } q(w) \| f(x_1, \ldots, x_n) \\ \text{and } |q(w)| \geq k}} \mkern-40mu \ket{i,p,w,x_1,\ldots,x_n}\bra{i,p,w,x_1,\ldots,x_n}
\end{equation}
The probability that the circuit produces a correct partial assignment of at least $k$ output values is given by $\norm{\Pi_{k} \ket{\psi_T}}^2$.
For a given partial assignment $q(w)$ to some outputs, we can define $\Pi_{q(w)}$ to be the projection onto the values of $\ket{X}$ where $q(w) \| f(X)$.
More specifically we have that:
\begin{equation}\label{eq:output-project}
    \Pi_{q(w)} = \mkern-40mu \sum_{\substack{x_1, \ldots, x_n\\ \text{s.t. }q(w)\| f(x_1,\ldots,x_n})} \mkern-40mu \ket{x_1,\ldots, x_n} \bra{x_1, \ldots, x_n}
\end{equation}
By construction when $q$ always produces a partial assignment of at least $k$ elements we have that
\begin{math}
\Pi_k = \sum_{i,p,w} \ket{i,p,w}\bra{i,p,w} \otimes \Pi_{q(w)}.
\end{math}

\paragraph{Space-bounded quantum computation}
As described above, we think of space-bounded quantum circuits as starting in the all $\ket{0}$ state and cycling between applying input queries $\calO$, and arbitrary input-independent computation $U_t$.
Unlike in the unitary circuit model, we allow our space-bounded quantum circuits to make intermediate measurements after applying each $U_t$ as shown in \cref{fig:quantum-circuit}.
Adopting the notation of \cite{BK23}, we will consider the set of consecutive $\calO$, $U_t$ and measurement gates as layer $L_t$.
As was done in \cite{HM21}, we will assume that the quantum query circuit has a dedicated register containing a boolean flag and a potential output $(i,y_i) \in [m] \times \range$.
After each query $\calO$ and subsequent unitary operation $U_t$, the flag register is measured in the standard basis.
Should the outcome $1$ be obtained, the output register is measured in the standard basis and interpreted as an output pair $(i,y_i)$ which is written to a write-only tape.
Otherwise, the circuit produces no output during this layer.
The space of layer $L_t$ is the number of qubits that are passed from layer $L_t$ to $L_{t+1}$ and is denoted $S_t$.
We define the space of a circuit as the maximum space of any layer, the time as the total number of layers, and the cumulative memory as the sum over all the $S_t$.
Thus the space needed to store the input and output is not included in this model.

Intermediate measurements enable circuits to produce parts of their output early and discard unnecessary ancillary qubits.
Similar to the disjoint collisions bound in \cite{HM21}, our results in \cref{sec:mat-vec,sec:mat-mult} apply to quantum circuits without any required structure on their output order.
Thus, as long as the circuits produce the correct output value for each index $i$, they may do so during arbitrary layers of the circuit that may depend on the chosen input.
However, as was the case in \cite{KSdW07}, our results for  quantum Boolean matrix multiplication in \cref{sec:quantum-boolean} apply to a more restricted model of computation where the choice of when to produce each output value is independent of the input.
In this \emph{output-oblivious} model, quantum circuits do not have a flag register.
Instead, on predefined layers the quantum circuit measures the output register in the standard basis and interprets the result as an element of $\range$ corresponding to a fixed output index.
This output-oblivious ordering restricts the set of allowed algorithms and is necessary to prove our key lemmas associated with Boolean matrix multiplication.

\begin{figure}
    \centering
    \includegraphics[width=.65 \textwidth]{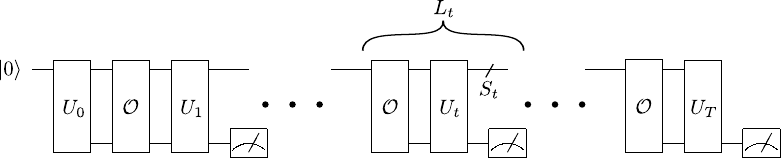}
    \caption{A general quantum circuit with $T$ queries.}
    \label{fig:quantum-circuit}
\end{figure}

\paragraph{Space-bounded classical computation}
One can view our classical lower bounds in \cref{sec:quantum-boolean} as applying to layered \emph{branching programs}~\cite{BC82}  where the space bound corresponds to the logarithm of the width of the program and the time corresponds to the number of layers.
Output in a branching program is produced along the edges and written to a write-only output tape.
Thus the space bound of a classical computation only considers the $S$ bits of internal state maintained by the device and not the size of its read-only input or write-only output.
Our results for classical Boolean matrix multiplication in \cref{sec:quantum-boolean} apply to an output-oblivious model, which corresponds to branching programs that must produce outputs for the same output index regardless of which edge is taken between two layers.

\paragraph{The Borodin-Cook method}
The Borodin-Cook method provides a general framework for proving time-space tradeoff lower bounds for multi-output problems, those for which every input
vector in $\domain^n$ is associated with some fixed set of possible output values from set $\range$ and the objective is to output at least $m$ of these output values.
As discussed earlier these can be functions
$f:\domain^n\rightarrow \range^m$, or
$f:\domain^n \rightarrow \mathcal{P}(\range)$ where the requirement is
to produce at least $m$ elements
of $f(\domain^n)$, if they exist\footnote{There is a more general version where the query algorithm is only required to produce these $m$ outputs with some sufficiently high probability but we focus on the simpler form}.

The property of the function $f$ that
enables the Borodin-Cook method to be used is the following\footnote{We do not specify an upper limit on the possible $k\le m$ in this informal statement. The exact range for which it holds will impact the space bounds for which the tradeoff holds.} for some well-behaved function $h(k,n)$:
\begin{description}
    \item[(*)] 
    Let $c=c(\domain)>1$.
Any classical query algorithm that makes at
most $t\le h= h(k,n)$ queries for an input distribution $\mathcal{D}$ on
$\domain^n$, correctly produces $k$ correct output values of $f$ with
probability at most $c^{-k}$.
\end{description}
With this property, Borodin and Cook showed that one directly obtains a classical time-space tradeoff for computing $f$ of the form $T\cdot S=\Omega(m \, h(S/(\log c),n)\, \log c)$ for time $T$ that is $n^{O(1)}$ and space $S$ as follows:
\begin{proof}[Proof sketch]
Choose $k$ with $\log n\le k\le m$ such that $2^S\cdot T\cdot c^{-k}< 1$;
then $k$ is roughly $S/(\log c)$.

Divide the $T$ query steps into disjoint blocks of
$h=h(k,n)$ queries each and assume that $T$ is a multiple of $h$, without loss of generality.
Since $m$ outputs must be produced on
all inputs in $\domain^n$ and there
are $T/h$ blocks,
for $T< m h/k$, which is $\Theta(m\, h \log c\ /S)$, for every execution on every input there must some
block where at least $k$ correct outputs are produced.  

However, since the space is at most
$S$ there are at most $2^S$ configurations of the states that the algorithm could have been in at the beginning of each time block.  Since \textbf{(*)} says that any fixed block can produce at least $k$ output values correctly
with probability at most $c^{-k}$ 
under $\mathcal{D}$, by a union bound the total probability that some fixed block
produces at least $k$ correct output values
is at most $2^S  c^{-k}<1/T$ by our
choice of $k$.
Since there are only $T/h$ blocks, the
probability that there is one of them that produces
$k$ correct answers is $<1$.

Therefore $T$ must be $\Omega(mh\log c\ /S)$ as required.
\end{proof}

For quantum algorithms, Klauck et al.~\cite{KSdW07} observed that
one could use a result by Aaronson in place of the union bound over the $2^S$ classical
state configurations at the start of each block in the Borodin-Cook method.

\begin{proposition}[\cite{Aar05}]\label{prop:quant-union}
Let $\calC$ be a quantum circuit, $\rho$ be an $S$-qubit (possibly mixed) state, and $\pi_{\textrm mix}$ be the $S$-qubit maximally mixed state. If $\calC$ starting in initial state $\rho$ produces some output $z$ with probability $p$, then $\calC$ starting in state $\pi_{\textrm mix}$ will produce $z$ with probability $q$ which is at least $p/2^{S}$.
\end{proposition}
We include a stand-alone derivation here for completeness.
\begin{proof}
    Without loss of generality we can assume $\calC$ performs no measurements until the end of the circuit.
    Thus we can think of $\calC$ as representing a unitary operator $U$.
    Let $\Pi_z$ be the projection onto output states of $\calC$ that cause the circuit to output the value $z$.
    Then $p_z =\text{Tr}[ \Pi_z U \rho U^\dagger]$.
    By the spectral decomposition theorem we can represent $\rho$ as a convex combination of some set of orthogonal pure states $\rho = \sum_{i \in [2^S]} \lambda_i \ket{\varphi_i}\bra{\varphi_i}$.
    Since the maximally mixed state can be represented as $\pi_{\textrm mix} = \sum_{i\in[2^S]} (1/2^S) \ket{\varphi_i}\bra{\varphi_i}$ we have that:
    \begin{align*}
        q &= \text{Tr}[ \Pi_z U \pi_{\textrm mix} U^\dagger]\\
        &= \text{Tr}[ \Pi_z U \left(\sum_{i\in [2^S]} \frac{1}{2^S} \ket{\varphi_i}\bra{\varphi_i}\right) U^\dagger]\\
        &= \frac{1}{2^S} \text{Tr}[\sum_{i\in 2^S} \bra{\varphi_i} U^\dagger \Pi_z U \ket{\varphi_i}]\\
        &\geq \frac{1}{2^S} \text{Tr}[\sum_{i\in 2^S} \lambda_i \bra{\varphi_i} U^\dagger \Pi_z U \ket{\varphi_i}]\\
        &= \frac{1}{2^S} \text{Tr}[\Pi_z U \left(\sum_{i \in [2^S]} \lambda_i \ket{\varphi_i}\bra{\varphi_i}\right) U^\dagger]\\
        &= \frac{1}{2^S} \text{Tr}[\Pi_z U \rho U^\dagger] = p/2^S
    \end{align*}
    Where the inequality comes from the fact that $\bra{\varphi} U^\dagger \Pi_z U \ket{\varphi} \geq 0$ for any state $\ket{\varphi}$.
\end{proof}

With this they showed that essentially the same paradigm could be used to give similar time-space tradeoff lower bounds for quantum algorithms if one can
prove a quantum analog of \textbf{(*)}.
One subtlety that arises from the quantum version of the Borodin-Cook method is that often the quantum version of \textbf{(*)} is proven in a non-space-bounded unitary circuit model without intermediate measurements.
By using the deferred measurement principle, we can see that lower bounds on the success probability of short quantum circuits in this model imply equally tight lower bounds in the space-bounded model where we directly apply the Borodin-Cook method.}

\makeatother
\makeatletter

\pgfplotsset{compat=1.18} 

\begin{document}

\title{Quantum Time-Space Tradeoffs for Matrix Problems
}

\author{
Paul Beame\thanks{Research supported by NSF grants CCF-2006359 and CCF-2422205}\\Computer Science \& Engineering\\University of Washington \and Niels Kornerup\thanks{Research supported by Schmidt Sciences Polymath award to David Soloveichik and the LDRD Program at Sandia National Laboratories. Sandia is managed and
operated by NTESS under DOE NNSA contract DE-NA0003525}\\Sandia National Laboratories
\and Michael Whitmeyer\thanks{Research supported by NSF grants CCF-2006359 and CCF-2422205 in addition to Simons Foundation grant  928589}
\\Computer Science \& Engineering\\University of Washington
}
\date{\today}

\maketitle
\thispagestyle{empty}

\begin{abstract}
We consider the time and space required for quantum computers to solve a wide variety of problems involving matrices, many of which have only been analyzed classically in prior work.
Our main results show that for a range of linear algebra problems---including matrix-vector product, matrix inversion, matrix multiplication and powering---existing classical time-space tradeoffs, several of which are tight for every space bound, also apply to quantum algorithms with at most a constant factor loss.
For example, for almost all fixed matrices $A$,
including the discrete Fourier transform (DFT) matrix, we prove that quantum circuits with at most $T$ input queries and $S$ qubits of memory require $T=\Omega(n^2/S)$ to compute matrix-vector product $Ax$ for $x \in \{0,1\}^n$. We similarly prove that matrix multiplication for
$n\times n$ binary matrices
requires
$T=\Omega(n^3 / \sqrt{S})$.
Because many of our lower bounds are matched by deterministic algorithms with the same time and space complexity, our results show that quantum computers cannot provide any asymptotic advantage for
these problems with any space bound.

We obtain matching lower bounds for the stronger notion of quantum cumulative memory complexity---the sum of the space per layer of a circuit.

We also consider Boolean (i.e. AND-OR) matrix multiplication and matrix-vector products, improving the previous quantum time-space tradeoff lower bounds for $n\times n$ Boolean matrix multiplication to $T=\Omega(n^{2.5}/S^{1/4})$ from $T=\Omega(n^{2.5}/S^{1/2})$.

Our improved lower bound for Boolean matrix multiplication is based on a new 
coloring argument that extracts more from the strong
direct product theorem that was the basis for prior work.
To obtain
our tight lower bounds for linear algebra
problems, we require much stronger bounds
than strong direct product theorems.   We obtain these bounds
by adding a new bucketing method to the quantum recording-query technique of Zhandry that lets us apply classical arguments to upper bound the success probability of quantum
circuits.
\end{abstract}
\newpage\setcounter{page}{1}
\section{Introduction}

Matrix computations are among the most fundamental
computational problems and are critically important in areas such as numerical and scientific computing, optimization, and machine learning.
If quantum computers can be shown to have a significant advantage over classical computations for these types of problems then it would open up a wide range of applications for such devices.

Prior work has shown that non-standard versions of matrix problems may indeed admit exponential or large polynomial quantum advantage:
For any efficiently implementable operator $M$, the HHL algorithm of Harrow, Hassidim, and Lloyd~\cite{HHL09} (with the improvements of \cite{CKS15}) 
can efficiently $\epsilon$-approximate the value of $x^\dagger M x$ for the
solution $x$ of a well-conditioned linear system.
However, it is important to note that this algorithm requires the input to be presented in an unconventional format.

Many extensions of the HHL algorithm have also been proposed that can be elegantly described in the quantum singular value transform (qSVT) framework first described in \cite{LC19} and popularized by \cite{GSLW19}.  
Despite initial hope of exponential speed-up, a series of papers by Tang and co-authors, and others~(e.g. \cite{Tang19, CGA+20-2,CGA+20,GST22, BT23,DBLP:conf/icml/ChepurkoCHLW22})
has shown that, by providing classical algorithms a comparable input format to the HHL algorithm, these quantum algorithms can be replaced by classical ones with only a polynomial blowup in the running time, 
although this polynomial is not always small.

This body of work still begs the question: What is the conventional quantum complexity of standard
classical problems like explicitly computing linear-system solutions, multiplying or inverting matrices, computing matrix-vector products, and computing the low rank approximation of a matrix?

By the polynomial method, we know that computing a single inner product (or parity) of
$n$-bit vectors requires $\Omega(n)$ quantum queries~\cite{BBC+01}, but linear algebra computations generally involve $\Omega(n)$ or 
$\Omega(n^2)$ such computations.
Sherstov~\cite{DBLP:journals/siamcomp/Sherstov12}, generalizing results of Klauck, \v{S}palek, and de Wolf~\cite{KSdW07} for the OR function,
gave a strong direct product lower bound for
quantum query complexity proved using the polynomial
method, which yields strong lower bounds for 
inner products involving many \emph{disjoint} input vectors.   However, the matrix problems in linear 
algebra are very far from direct product problems: 
The vectors involved are highly correlated with
each other, so this prior work does not shed light on the key question of whether quantum 
algorithms provide any advantage for general linear algebra.

In this paper, we resolve
these questions for quantum computation of a
wide array of linear algebra problems,
proving lower bounds for
quantum computation
that are asymptotically the same as the best
classical 
lower bounds. 
Since many of the problems also have 
deterministic algorithms whose resource
usage matches the lower bounds, our results show
that there is provably no asymptotic quantum advantage at all in solving these linear algebra problems!

As with the study of classical computation involving super-linear time lower bounds, we consider quantum algorithms in which we 
limit the number of qubits of memory and hence produce quantum time-space tradeoffs.
That is, for each fixed bound on the amount
of memory allowed, we derive asymptotically 
the same time lower bound for the quantum 
algorithm as one would get for the time lower
bound on classical algorithms with the same number of
classical bits.
In many ways, quantum memory is an even more critical resource than classical memory since it is a measure of the maximum number of qubits that maintain coherence at any time during the algorithm's execution.
For this reason the first general-purpose fault-tolerant quantum computers will likely have very limited memory and only be able to execute low depth quantum circuits.
As such, it is crucial to consider both the time and space complexity for quantum algorithms.

We prove our lower bounds for quantum computation in a query model where algorithms are able to perform arbitrary input-independent unitary transformations on their state between quantum queries to their input.
This is a sufficiently general model that our lower bounds also apply to any reasonable model of quantum computation---including quantum circuits where the (classical) input is stored in quantum-readable read only memory (QROM).

The keys to proving our time-space tradeoffs are new results proving much stronger lower bounds than strong direct product theorems for matrix-vector products and matrix multiplication. While our bounds have the same form as strong
direct product theorems (the success probability decays exponentially with the number of outputs), they also apply with almost completely overlapping sets of inputs, in contrast to the disjoint
inputs that are necessary to apply direct product theorems.

While there is a large body of work proving strong classical
time-space tradeoffs~(e.g. \cite{Tom78, BFK+79, Yes84, BC82, DBLP:conf/focs/Abrahamson90,Abr91, DBLP:journals/siamcomp/Beame91,DBLP:journals/tcs/MansourNT93}) and a 
large body of work analyzing unrestricted quantum query algorithms versus their classical randomized counterparts~(e.g. \cite{DJ92, BV97, Sim97, BBC+01, Amb02, SS06, Spa08, She11}),  there are just a few previous papers that
analyze the quantum memory required to make use of these quantum queries.
Klauck, \v{S}palek, and de Wolf~\cite{KSdW07} extended the classical method of Borodin and Cook~\cite{BC82} for
proving time-space tradeoffs to quantum circuits
using a new strong direct product theorem for quantum query algorithms computing the OR function. 
They showed that algorithms making $T$ quantum
queries and using $S$ qubits of quantum memory require $T=\Theta(n^{1.5}/S^{1/2})$ to sort lists of length $n$,  and 
require $T=\Omega(n^{2.5}/S^{1/2})$
to compute $n\times n$ Boolean matrix product. 
Ambainis, \v{S}palek, and de Wolf~\cite{ASdW09} extended this direct
product approach to 2-sided error algorithms computing $k$-threshold functions which allowed them to produce similar trade-off lower bounds for systems of linear inequalities/equalities (though these have the drawback, unlike the other results, that the hard function for space $S$ depends on the space bound).  This approach, based on an extension of the adversary method using eigenspace analysis, was very difficult to apply.

As a result, further study of quantum time-space tradeoff lower bounds languished until it was enabled by an idea of Zhandry~\cite{Zha19} who, motivated by understanding quantum algorithms interacting with random function oracles, developed an approach to understanding quantum query algorithms using a \emph{compressed oracle} and Fourier analysis.
This views computations in a \emph{recording query} basis that allow one to keep track of a quantum query algorithm as a superposition of basis states that have a natural classical query interpretation.  It has 
been applied to finding multi-way collisions~\cite{LiuZ18} and to inverting a random permutation~\cite{Rosmanis21}.  
This greatly simplifies the analysis of quantum query algorithms and can be applied to many lower bound methods that use randomly chosen inputs rather than being limited to cryptographic applications.

Extending Zhandry's approach, Hamoudi and Magniez~\cite{HM21} applied
an even cleaner expression of the method, using phase oracles with the recording query basis rather than Fourier analysis, and extended it using biased random inputs to
derive query lower bounds in a regime of exponentially small 
success probability.  They used this to obtain time-space tradeoff lower bounds, proving that any quantum algorithm that finds $K$ disjoint collisions in an input of length $n$ with $T$ quantum queries and $S$ qubits of memory must have
$T=\Omega(KN^{1/3}/S^{1/3})$.  They also re-derived the earlier sorting lower bound using this method.

\paragraph{Our linear algebra lower bounds and methods}

Time-space trade-off lower bounds
for linear algebraic problems were
among the first to be studied for classical
computation~\cite{Yes84} after the first bounds
for sorting.  
The strongest classical results are due
to Abrahamson~\cite{Abr91} who developed a powerful general method based on matrix rigidity. 
This yields state-of-the-art lower bounds for computation of Fourier transforms, convolution, matrix-vector products, matrix multiplication,  matrix inversion, matrix powering, and linear system solving.
The lack of any analogous results for quantum
computation has been a substantial gap in our understanding~\footnote{Over a field of $>n$ elements one can reduce $n\times n$ Boolean matrix multiplication to ordinary multiplication of $0$-$1$ matrices but the lower bound is inherently too weak because in the Boolean case each output bit is a disjointness function of its inputs and hence can be computed using only $O(\sqrt{n})$ quantum queries using 
Grover's algorithm (\cite{Gro96}).}.   

Our results show that all the linear algebraic time-space tradeoff lower bounds shown by Abrahamson~\cite{Abr91} also apply
to quantum computation even when the quantum circuit can adaptively decide when to produce output based on the observed input.
Since many of these classical lower bounds are tight, our results directly imply that there is no hybrid classical-quantum algorithms with a polynomial advantage for these problems unlike the query bounds for search and collision finding in \cite{HamoudiLS22}.
Using the generic results in \cite{BK23}, we also prove asymptotically equivalent lower bounds on the stronger notion of quantum cumulative memory complexity for these problems.
We include a table of our time-space tradeoff lower bounds in \cref{table:results}.

{\small
\begin{table}
\begin{tabular}{@{}lll@{}}
\toprule
Problem                                & Quantum  Lower Bound          & Source     \\ \midrule
Matrix Multiplication $f(A,B) = AB$             & $T=\Theta(n^3 \sqrt{\log d\ / S})$ & \cref{thm:mat-mult}       \\ \midrule
Matrix Squaring $f(A) = A^2$                   & $T = \Theta(n^3 \sqrt{\log d\ / S})$ & \cref{cor:mat-square} \\ \midrule
Matrix Triple Product $f(A,B,C) = ABC$                       & $T=\Theta(n^4 \log d\ / S)$        &  \cref{cor:mat-mat-mat-product-ts-lb}       \\ \midrule
Matrix Cubing $f(A) = A^3$                         & $T=\Theta(n^4 \log d\ / S)$         & \cref{cor:mat-cube-ts-lb}       \\ \midrule
Matrix Inversion $f(A) = A^{-1}$                     & $T=\Omega(n^4 \log d\  / S)$        & \cref{cor:mat-invert-ts-lb}       \\ \midrule
System of Linear Equations $f(A, y) = A^{-1} y$                 & $T=\Omega(n^3 \log d\  / S)$        & \cref{cor:system-eqn-cm-lb}       \\ \midrule
Matrix-Vector Product $f(x) = Ax$          & $T=\Theta(n^2 \log d\ /S)$           & \cref{thm:time-space-matrix-vector}       \\ \midrule
Discrete Fourier Transform $f(x) = Wx$             & $T=\Theta(n^2 \log d\ /S)$           & \cref{cor:DFT-ts-lb}       \\ \midrule
Convolution $f(u,v) = u*v$                           & $T=\Theta(n^2 \log d\ /S)$           & \cref{cor:convo-ts-lb}       \\\midrule
Binary Integer Multiplication                & $T=\Omega(n^2/(S \log^2 n))$        &  \cref{cor:bin-mult-ts-lb}       \\\midrule \midrule
Boolean Matrix Multiplication $f(A,B) = A \bullet B$   & $T=\Omega(n^{2.5} / S^{0.5})$          & \cite{KSdW07}   \\
                                       & $T=\Omega(n^{2.5}/ S^{0.25})$            & \cref{thm:quantum-booolean-matrix}       \\ 
                      \hfill Classical & $T = \Omega(n^3 / S)$                   & \cite{KSdW07,DBLP:conf/focs/Abrahamson90} \\
                     \hfill Classical & $T = \Omega(n^{3.5} / S)$  for $S\ge cn$                  & \cite{DBLP:conf/focs/Abrahamson90} \\
                      \hfill Classical & $T = \Theta(n^3 / S^{0.5})$             & \cref{thm:classical-boolean-matrix} \\\midrule
Boolean Matrix Squaring & $T= \Omega(n^{2.5} / S^{0.25})$ & \cref{cor:bool-matrix-sq} \\ \midrule
\end{tabular}
\caption{Summary of our quantum lower bounds, along with prior work.
Inputs are assumed to be of length $n$ vectors or $n \times n$ matrices.
Our linear algebra bounds apply for input elements coming from any fixed subset $\domain$ of a field with $d= |\domain|$. These are the first quantum time-space lower bounds for all of these problems other than Boolean matrix multiplication. Problems with deterministic classical query algorithms given in \cite{JS82} and \cite{Abr91} that match our quantum query lower bounds are denoted with $\Theta$ notation instead of $\Omega$. Constructions of the matching query algorithms can be found in \cref{sec:query-algs}.}
\label{table:results}
\end{table}
}

As discussed already, we need a much stronger lower bound method than any derivable from strong
direct product theorems.
We do this by the adding new ideas to the compressed oracle/recording query approach of Zhandry~\cite{Zha19} as extended and applied by
Magniez and Hamoudi~\cite{HM21}.
Thus far, the compressed oracle method has used a two-step pattern:   First, identify a notion of unusual progress of a quantum algorithm towards a solution  (i.e., the partial information so far is more determinative of the answer than one might expect) and show that the
total amplitude of states where this occurs is small,
Second, show that the total amplitude of the quantum states where many outputs are produced without
unusual progress can be bounded; this latter part
has used ideas with classical analogs that can 
be applied by breaking the algorithm's final state into mutually orthogonal components, each with small amplitude on the correct answers.   

However, in our case with linear algebra problems, there
is no form of unusual progress and also no clear way to break up the problem into mutually orthogonal basis states.    
Thus, neither part of the pattern seems to work.
Instead, we can use the recording query framework to characterize how much a quantum circuit can know about its input.
We use the triangle inequality to bucket amplitude from the algorithm's state into a small number of non-orthogonal components (or buckets) that share some set of inputs that they know nothing about.
We can then apply a classical argument showing that each component must have small amplitude on the correct answers.
By finding a way to divide the state into a small number of buckets that each have small amplitude on correct answers, we can obtain tight lower bounds.
The properties required of this division become more subtle as we move
to the problem of matrix multiplication, where in order to get small amplitude, we need to contend with a partition featuring significantly more parts.

\paragraph{Improved bounds for Boolean matrix operations}

Here we improve the previous lower bound for quantum algorithms computing 
Boolean matrix multiplication given in~\cite{KSdW07} from
$T=\Omega(n^{2.5}/S^{1/2})$ to $T=\Omega(n^{2.5}/S^{1/4})$.
We do this using a more sophisticated embedding of the $k$-fold direct product
of OR functions into an arbitrary subset of $k$ outputs of Boolean matrix multiplication.   The embedding hinges on the number of colors needed for a certain kind of
partial coloring of subsets $E$ of the $n\times n$ grid.
The exponents of $n$ and $S$ in our lower bound are optimal for the general quantum circuit model to which it applies. 

Our lower bounds also lead to  improving the classical lower bound
tradeoff of $T=\Omega(n^3/S)$ for circuits shown in~\cite{KSdW07} to $T=\Omega(n^3/S^{1/2})$.   (In these bounds, $T$ is circuit depth
and $S$ is
circuit width.)
Just as with our quantum lower bound, this has optimal exponents for $n$ and $S$, achieving the goal of 
Klauck, \v{S}palek, and de Wolf~\cite{KSdW07} who suggested that $T^2 S=\Omega(n^6)$ was a likely tight tradeoff for classical computation of Boolean matrix multiplication.
It is strictly larger almost everywhere than a classical lower bound of $T=\Omega(n^3/S)$ 
for $S\le n^{0.5}$ and $T=\Omega(n^{3.5}/S)$ for $S\ge n$ for Boolean matrix multiplication on branching programs (a more general model than circuits) due to 
Abrahamson~\cite{DBLP:conf/focs/Abrahamson90} that is tight almost surely
for input matrices whose entries are 1 with probability $1/\sqrt{n}$
independently.

Finally, we make a small adjustment to convert the Boolean matrix-vector lower bounds and lower bounds for systems of inequalities given in
\cite{KSdW07} and \cite{ASdW09}, respectively, so that the
problems that are shown hard for space $S$ do not depend on $S$.

\subsection{The quantum recording query technique}
Here we review the methods developed in \cite{Zha19, HM21} that allow us to analyze what a quantum circuit learns about its input by making quantum queries.
We will assume that the input state $\ket{\psi_0}_\calO$ is the equal superposition state over all inputs, although \cite{Zha19, HM21, Rosmanis21} generalize this method to other input distributions.
We can exchange the general query operator $\calO$ for the uniform input distribution with a recording query operator $\calR$ that we define as follows:

\begin{definition}[adapted from \cite{HM21}]\label{def-recording-query-oracle}
Let $\domain$ be the input alphabet, $d= |\domain|$, and $\nu$ be our choice of canonical bijection between $\domain$ and $\{0,\ldots, d-1\}$.
We define $\calS_1$ to be the unitary operator that maps
\begin{displaymath}
    \calS_1: \begin{cases} \ket{\perp} &\longrightarrow \frac{1}{\sqrt{d}} \sum_{y \in \domain} \ket{y} \\ \frac{1}{\sqrt{d}} \sum_{y \in \domain} \ket{y} &\longrightarrow \ket{\perp} \\ \frac{1}{\sqrt{d}} \sum_{y \in \domain} \omega_{d}^{p\,\nu(y)} \ket{y} &\longrightarrow \frac{1}{\sqrt{d}} \sum_{y \in \domain} \omega_{d}^{p\,\nu(y)} \ket{y} \; \forall p \in \{1, \ldots, d -1\}.\end{cases}
\end{displaymath}
Let $\calS = (I)_{i,p,w}\otimes (\calS_1^{\otimes n})_{x_1, \ldots, x_n}$ and $\calO$ be the standard oracle operator that maps the basis state 
\begin{displaymath}
    \ket{i,p,w,x_1, \ldots, x_n} \longrightarrow \omega_{d}^{p\,\nu(x_i)} \ket{i,p,w,x_1,\ldots,x_n}.
\end{displaymath}
Then the \emph{recording query oracle operator} $\calR$ is defined as $\calS \calO \calS$.
\end{definition}
$\calS_1$ introduces $\perp$ as a new value for the input registers.
Intuitively, the $\perp$ symbol indicates that the algorithm does not know anything about that register of the oracle.
Hence by adding and correctly manipulating the $\perp$ symbols in the oracle's registers, we can record what the algorithm knows about the input. Since $\calS^2 = I$, we can exactly characterize how the states of quantum circuits with oracles $\calO$ and $\calR$ relate to one another.

\begin{proposition}[Theorem 3.3 in \cite{HM21}]\label{prop:recording-equiv}
Let $\calC$ be a quantum circuit that for each $j\le t$ applies unitary $U_j$ after the $j$-th query. Let $\calS$ be the unitary operation and $\calR$ be the recording query oracle from \cref{def-recording-query-oracle}. Let
\begin{align*}
    \ket{\psi_t} &= U_{t} \calO U_{t-1} \ldots U_1 \calO U_0 \left(\ket{0}_{i,p,w} \otimes \frac{1}{d^{n/2}} \sum_{x_1, \ldots, x_n \in \domain} \ket{x_1, \ldots, x_n}_{x_1, \ldots, x_n}\right)\\
    \ket{\phi_t} &= U_t \calR U_{t-1} \ldots U_1 \calR U_0 \left(\ket{0}_{i,p,w} \otimes \ket{\perp}_{x_1, \ldots, x_n} \right)
\end{align*}
be the states of $\calC$ with oracle $\calO$ or $\calR$ respectively. Then $\ket{\psi_t} = \calS \ket{\phi_t}$.
\end{proposition}

In other words, it is impossible to distinguish the final state $\ket{\psi_T}$ of a circuit with standard oracle $\calO$ from the output with recording oracle $\calR$ if we apply $\calS$ to the registers of $\calR$ after the final query.
Thus we can conclude that the success probability of a quantum circuit with $T$ queries producing a partial assignment of $k$ correct output values is given by $\norm{\Pi_{k} \ket{\psi_T}}^2 = \norm{\Pi_{k} \calS \ket{\phi_T}}^2$.
Note that while $\ket{\phi_T}$ may have inputs in the $\perp$ state, \cref{prop:recording-equiv} tells us that $\calS \ket{\phi_T}$ will never have an input in the $\perp$ state.
This means that when considering recording query oracles, it is safe to keep our current definitions of $\Pi_{k}$ and $\Pi_{q(w)}$ which will always project out any basis state where an input is assigned to $\perp$.
We will leverage the following property of $\ket{\phi_T}$ to bound the success probability of quantum circuits with at most $T$ queries.

\begin{definition}
Let $\Gamma_t$ be the set of all elements
$(\domain\cup \{\bot\})^n$ with at most $t$ non-$\bot$ 
elements.   This is the set of indices for all
recording query basis states associated with
quantum algorithms that make at most $t$ queries.
\end{definition}

\begin{proposition}[Fact 3.2 in \cite{HM21}]\label{prop:max-queried}
The state $\ket{\phi_t}$ from \cref{prop:recording-equiv} is a linear combination of basis states $\ket{i, p, w, x_1, \ldots, x_n}$ where $(x_1, \ldots, x_n) \in \Gamma_t$.  
\end{proposition}

For the bounds in \cite{HM21} it is essential to bound how the state of $\ket{\phi}_\calO$ can change after each query. For our use of the recording query technique, this detailed analysis is not necessary. 
Nevertheless, we state the following proposition here for completeness.

\begin{proposition}[Lemma 4.1 in \cite{HM21}]
\label{prop:query-approx}
Let $\domain$ be the input alphabet, $d= |\domain|$, and $\nu$ be our choice of canonical bijection between $\domain$ and $\{0,\ldots, d-1\}$.
If the recording query operator $\mathcal{R}$  is applied to a
basis state $\ket{i,p,w,x_1,\ldots,x_n}$  where $p \ne 0$ then the register $\ket{x_i}$ is mapped to
\begin{equation}
    \begin{cases}
        \sum_{y\in D}\frac{\omega_d^{p\, \nu(y)}}{\sqrt{d}}\ket{y}&\textrm{if }x_i=\bot\\
        (1-\frac{2}{d})\omega_d^{p\, \nu(x_i)}\ket{x_i} + \frac{1}{d} \ket{x_i}+ \frac{\omega_d^{p\, \nu(x_i)}}{\sqrt{d}} \ket{\bot}+ \sum_{y\in D\setminus \{x_i\}} \frac{1-\omega_d^{p\, \nu(y)}-\omega_d^{p\, \nu(x_i)}}{d}\ket{y}&\textrm{otherwise.}
    \end{cases}
\end{equation}
If $p = 0$ then the register remains unchanged.
\end{proposition}

\section{Quantum matrix vector products}\label{sec:mat-vec}
In this section, we consider the task of --- for a fixed matrix $A \in \F^{m \times n}$ --- computing the function $f_A(x) = Ax$ for inputs $x \in \domain^{m}$ (where $\domain$ is a fixed subset of $\F$) using a quantum circuit.
We note that this is a fundamentally harder task than is considered in many quantum machine learning papers (for example \cite{HHL09}) as we require the circuit to output a classical vector $y \in \F^{n}$ rather than either a quantum state encoding the entries of $y$ in the amplitudes or an estimate of $y^\dag M y$.
Also unlike many prior quantum time-space tradeoffs, including sorting \cite{KSdW07, HM21, BK23} and boolean matrix multiplication \cite{KSdW07} (and our \cref{thm:quantum-booolean-matrix}), our matrix vector product and matrix multiplication lower bounds apply to circuits that can adaptively decide when to produce each output based on the observed inputs.
Time-space lower bounds against such quantum circuits were first described in \cite{HM21} for the multiple disjoint collisions problem, although they were not able to show such a result for sorting.
Similar to \cite{HM21} we are able to lower bound these circuits by identifying a single hard distribution over the inputs that applies to any set of outputs.

\begin{theorem}
\label{thm:time-space-matrix-vector}
    Let $m\le n^r$ for some constant $r$ and $2\le d\le n^n$.
    There is a constant $C>0$ such that the following holds:
    Let $A$ be an $m\times n$ matrix over a field $\F$ that is $(g,h,c)$-rigid.  Then any quantum circuit using time $T$ and space $S< \frac{c}{6(r+6)} g \log_2 d$ that computes the function $f_A: \domain^n \to \F^m$ for $\domain\subseteq \F$ with $d = |\domain|$ given by $f_A(x) = Ax$ with success probability larger than $2^{-S}$ requires that $T\ge C m h \log d \ / S$.
\end{theorem}

When the fixed matrix $A$ is sufficiently rigid, for example when both $g$ and $h$ are linear in $n$ as is the case with the DFT matrix per \cref{prop:DFT-rigid} or a random matrix with high probability per \cref{prop:rigid-matrices}, this lower bound becomes $\Omega(mn \log d)$ provided that $S$ is at most some
constant times
$n\log d$ which is essentially a trivial
constraint for the problem.
This bound is tightly matched by a classical query algorithm in \cref{prop:mat-vec-upper-bound-alg}.

This theorem follows from the following key lemma, proven in \cref{subsec:quant-mat-vec-lemma}, which lets us bound the number of correct output values produced by a shallow quantum circuit.

\begin{lemma} \label{lem:matvec}
    Let $A$ be any $(k,h,c)$-rigid $m \times n$ matrix over a finite field $\F$ and let $f_A:\domain^n \to \F^m$ for $\domain \subseteq \F$ be defined by $f_A(x) = Ax$. 
    Then for $\alpha > 0$ and for input $x$ sampled
    uniformly from $\domain^n$ and any quantum circuit $\calC$ with at most $\alpha h$ queries to $x$, the probability that $\calC$ produces  $k$ correct output values of $f_A(x)$ is 
    at most $\lceil h/(ck)\rceil^2\,  (4^{H_2(\alpha)}/|\domain|^{1-\alpha})^{ck}$.
\end{lemma}

\begin{sloppypar}
Note: For $\alpha\le 0.0737$ we have $1-\alpha-2H_2(\alpha)>1/6$ and 
hence the bound is at most $\lceil h/(ck)\rceil^2 |\domain|^{-ck/6}$ for $d\ge 2$.
\end{sloppypar}

\begin{proof}[Proof of \cref{thm:time-space-matrix-vector} from \cref{lem:matvec}]
    Let $\calC$ be a quantum circuit with $T$ queries and space $S$ that computes $f_A(x)$ with success probability larger than $2^{-S}$.
    Since $h\le n$, $m\le n^r$ and $S\ge \log_2 n$ we
    only need to consider the case that $T\le n^{r+1}\log_n d\le n^{r+2}$.
    
    
    Let $\alpha=0.0737$. We partition $\calC$ into $\ceil{T/(\alpha h)}$ sub-circuits that each have at most $\alpha h$ queries.
    By combining \cref{prop:quant-union} and \cref{lem:matvec}, we know that each sub-circuit can produce $k\le g$ correct output values with probability at most $2^{S} \ceil{h/(ck)}^2 d^{-ck/6}\le h^2\, 2^{S} d^{-ck/6}$.

    
   \begin{sloppypar}
By assumption, we have $d^{-cg/6}\le 2^{-(r+6)S}\le n^{-(r+4)} 2^{-2S}\le h^{-2} 2^{-2S}/T$ since $S\ge \log_2 n$, $T\le n^{r+2}$, and $h\le n$.
In particular, this implies that
$h^2 d^{-cg/6}< 2^{-S}$ so we must have $T>\alpha h$ by \cref{lem:matvec}.
Set $k\le g$ to be the smallest integer such that 
    $h^2\, 2^{S} d^{-ck/6}\le 2^{-S}/T$. 
    Then
  the probability that a sub-circuit produces $k$ correct output values is at most $2^{-S}/T$. 
This gives $k=\ceil{[6\log_{2}(h T) + 12S]/(c\log_2 d)}$.
We note that $k$ is at most $c^* S/ \log_2 d$ for some constant $c^*>0$
since $\log_2(hT)\le (r+3)\log_2 n\le (r+3)S$.  
   \end{sloppypar}
  
    Taking a union bound over the sub-circuits, the probability that any of them produces $k$ correct output values is at most $2^{-S}$. Since $f_A$ has $m$ outputs, this means that
    \begin{displaymath}
    \ceil{T / (\alpha h)} (k-1) \geq m
    \end{displaymath}
    Since $T\ge \alpha h$, we have
    \begin{displaymath}
    2Tk\ge \alpha mh.
    \end{displaymath}
    Plugging in our upper bound on $k$ we have that
    \begin{displaymath}
        2c^* T S/\log_2 d \geq \alpha mh
    \end{displaymath}
    and hence $T\cdot S$ is at least
    $\frac{\alpha}{2c^*}m h \log d$ as claimed.
\end{proof}

\subsection{Success probability of small depth quantum circuits}\label{subsec:quant-mat-vec-lemma}

We first give an overview of the argument, which involves an initial
uniform distribution over the inputs
$x\in D^n$.
This begins by decomposing the state after $t\le \alpha h$ queries into orthogonal components based on the values of working qubits $\ket{i,p,w}$, which also determine the set of $k$ output values produced.   It then
suffices to prove that for each fixed $\ket{i,p,w}$ the total fraction of the squared amplitude for any state that is spanned by recording query basis states with at most $t$ non-$\bot$ items can be on inputs for which the fixed output values are correct is 
exponentially small in $k$.  
  
If we knew which $t\le \alpha h$ input indices were queried, as we would with classical algorithms in the analysis of~\cite{Abr91}, then things would be easy:
Since the fixed matrix $A$ is $(k,h,c)$ rigid, the sub-matrix of $A$ with rows corresponding to these $k$ outputs, and with the $\ge n-\alpha h$ ``unqueried'' columns has rank at least $ck$, so any fixed output can be correct with probability at most $d^{-ck}$ over the choice of inputs.
However, the quantum state after $t$ queries is a superposition of
recording query basis states that could involve all possible subsets of $\le t$ non-$\bot$ indices which
is at least $\binom{n}{t}$ possibilities.   
We now give our proof in detail.

\begin{proof}[Proof of \cref{lem:matvec}]
Let $d = |\domain|$. For simplicity we will assume that $q(w)$---the output as a function of the measured value of the work register---always produces $k$ outputs.\footnote{If in general $q(w)$ produces more than $k$ outputs, we only consider its first $k$ outputs.}
Let $A$ be a $(k, h,c)$-rigid matrix.
By \cref{prop:max-queried}
 after $t\le \alpha h$ queries in the recording query oracle model,
 the state $\ket{\phi_t}$
 is a linear combination of basis states $\ket{i, p, w, x_1, \ldots, x_n}$ where $(x_1, \ldots, x_n) \in \Gamma_t$.
 It will be useful to be more 
 explicit in our discussion of $\Gamma_t$. 
 Each element of $\Gamma_t$ consists
 of an assignment $y\in \domain^I$
 for some subset $I\subseteq [n]$ with
 $|I|\le t$ and value $\perp$ on all
 coordinates in $[n]\setminus I$.
 Therefore,
 we can write the state as:
\begin{equation}
\ket{\phi_t} = \sum_{\substack{i,p,w \\ I \subseteq [n],\  |I| \leq t \\ y \in \domain^I}} \alpha_{i,p,w,I, y} \ket{i,p,w}\ket{y}_I \ket{\perp}_{[n] \setminus I}\label{eq:matrix-vector-1}
\end{equation}
for some $\alpha_{i,p,w,I, y}$ with $\sum_{i,p,w,I, y} |\alpha_{i,p,w,I, y}|^2=1$. Thus by \cref{prop:recording-equiv}, the final state of the algorithm (after $t \leq \alpha h$ queries) in the non-recording query oracle setting is given by:

\begin{displaymath}
\ket{\psi_t} = \calS \ket{\phi_t} = \calS \sum_{\substack{i,p,w \\ I \subseteq [n],\ |I| \leq t \\ y \in \domain^I}} \alpha_{i,p,w,I, y} \ket{i,p,w}\ket{y}_I \ket{\perp}_{[n] \setminus I}
\end{displaymath}
Since $\calS$ behaves as the identity on $\ket{\phi_t}_\calC$ and the $\ket{i,p,w}$ are orthogonal basis states, we can rewrite this as:

\begin{displaymath}\sum_{i,p,w} \beta_{i,p,w} \ket{i,p,w}  \otimes \bigg[\calS_1^{\otimes n} \sum_{\substack{I \subseteq [n],\ |I| \leq t \\ y \in \domain^I}} \beta^{i,p,w}_{I, y}\ket{y}_I \ket{\perp}_{[n] \setminus I}\bigg]\end{displaymath}
for some $\beta_{i,p,w}$ and $\beta^{i,p,w}_{I,y}$ such that
$\alpha_{i,p,w,I, y}=\beta_{i,p,w}\,\beta^{i,p,w}_{I,y}$, 
$\sum_{i,p,w} |\beta_{i,p,w}|^2=1$ and for each choice of $i,p,w$, we have that
$\sum_{I,y} |\beta^{i,p,w}_{I,y}|^2=1$.
With this decomposition, using the definition in~\cref{eq:output-k}, the success probability of producing $k$ correct output values is given by:
\begin{align*}
\norm[\big]{ \Pi_{k} \calS \ket{\phi_t}}^2 &= \norm[\bigg]{\Pi_{k} \sum_{i,p,w} \beta_{i,p,w} \ket{i,p,w}  \otimes \bigg[\calS_1^{\otimes n}\sum_{\substack{I \subseteq [n],\ |I| \leq t \\ y \in \domain^I}} \beta^{i,p,w}_{I, y}\ket{y}_I \ket{\perp}_{[n] \setminus I}\bigg]}^2\\
&=\norm[\bigg]{\sum_{\substack{i,p,w}} \beta_{i,p,w} \ket{i,p,w}  \otimes \bigg[\Pi_{q(w)} \calS_1^{\otimes n} \sum_{\substack{I \subseteq [n],\ |I| \leq t \\ y \in \domain^I}} \beta^{i,p,w}_{I, y}\ket{y}_I \ket{\perp}_{[n] \setminus I}\bigg]}^2
\end{align*}
where $\Pi_{q(w)}$ is defined as in \cref{eq:output-project} and 
is the projection of $\Pi_{k}$ onto fixed values of $q(w)$.
Since the basis states $\ket{i,p,w}$ are orthogonal and $\sum_{i,p,w} |\beta_{i,p,w}|^2 =1$, we have

\begin{equation}\norm[\big]{ \Pi_{k} \calS \ket{\phi_t}}^2\le \max_{i,p,w}  \norm[\bigg]{ \Pi_{q(w)} \calS_1^{\otimes n} \sum_{\substack{I \subseteq [n],\ |I| \leq t \\ y \in \domain^I}} \beta^{i,p,w}_{I, y} \ket{y}_I \ket{\perp}_{[n] \setminus I}}^2\label{eq:target}\end{equation}
We now fix $i,p,w$ and let $A_{q(w)}$ be the submatrix of $A$ restricted to the rows defined by the set of the $k$ output values $U$ associated with $q(w)$.
We can describe $\Pi_{q(w)}$
as a projection onto basis states $\ket{x_1 ,\ldots, x_n}$ such that:
\begin{displaymath}
    A_{q(w)} \begin{bmatrix} x_1 \\ \vdots \\ x_n \end{bmatrix} =  q(w).
\end{displaymath}

Since the basis states $\ket{y}_I\ket{\bot}_{[n]\setminus I}$ for distinct $I$ are
orthogonal in the recording query basis, they remain orthogonal in the standard basis after the $\calS$ operator is applied. However, the subsequent application of the $\Pi_{q(w)}$ projector makes these vectors no longer orthogonal.

To handle this, we bucket the sets $I \subseteq [n]$ with $|I| \leq t$ 
into a small number of buckets, $\calB_1,\ldots$, so that for each bucket $\calB_\ell$ we can bound:
\begin{equation*}
    \mu_\ell = \norm[\bigg]{\Pi_{q(w)}\calS_1^{\otimes n}\sum_{I\in \calB_\ell, y \in \domain^{I}}\beta^{i,p,w}_{I, y} \ket{y}_I \ket{\perp}_{[n] \setminus I}}^2
\end{equation*}
and then we can use the triangle inequality to bound the success probability as a sum of the $\mu_\ell$.

In particular, our key observation is that if a bucket of recording query basis states completely misses querying a fixed set of input variables that could completely scramble the value of a set of $r$ output values, then one cannot do better than randomly guess those output values.
More precisely, we show that the contribution to success from that bucket of basis states has amplitude
at most $\frac{1}{\sqrt{d^r}}$.

\begin{lemma}
\label{lem:group}
 Let $U\subseteq [m]$ be a set of output indices and $V\subseteq [n]$
 be a set of input indices with $|V|=|U|=r$ such that the submatrix $A_{U,V}$ is full rank.  
 Fix $q\in \F^U$ and define $\Pi_q$ to be the projection
 map onto the span of the set of basis states $\ket{x_1,\ldots,x_n}$ with $x_1\ldots x_n\in \domain$ such that 
 $A_U x=q$.
 Then for any collection $\calB$ of sets $I\subseteq [n]\setminus V$ and
 any quantum state 
 $\sum_{I\in \calB,\ y\in \domain^I} \eta_{I,y} \ket{y}_I \ket{\perp}_{[n] \setminus I}$ we have
 \begin{displaymath}
 \norm[\bigg]{\Pi_q\calS_1^{\otimes n}\sum_{I\in \calB,\ y\in \domain^I} \eta_{I,y} \ket{y}_I \ket{\perp}_{[n] \setminus I}}^2 \le \frac{1}{d^r}.
 \end{displaymath}
\end{lemma}

\begin{proof}
By definition each $I\in \calB$ satisfies $I\cap V=\varnothing$, so
\begin{align*}
& \Pi_q\calS_1^{\otimes n}\sum_{I\in \calB,\ y\in \domain^I} \eta_{I,y} \ket{y}_I \ket{\perp}_{[n] \setminus I} \\
 &= \Pi_q\calS_1^{\otimes n}(\ket{\bot}_V\otimes \sum_{I\in \calB,\ y\in \domain^I} \eta_{I,y} \ket{y}_I \ket{\perp}_{[n] \setminus (I\cup V)}\\
 &=\Pi_q(\calS_1^{\otimes r}\ket{\bot}_V\otimes \calS_1^{\otimes(n-r)}\sum_{I\in \calB,\ y\in \domain^I} \eta_{I,y} \ket{y}_I \ket{\perp}_{[n] \setminus (I\cup V)})\\
  &=\Pi_q(\sum_{y'\in \domain^V} \frac{1}{\sqrt{d^r}}\ket{y'}_V\otimes \calS_1^{\otimes (n-r)}\sum_{I\in \calB,\ y\in \domain^I} \eta_{I,y} \ket{y}_I \ket{\perp}_{[n] \setminus (I\cup V)})
 \end{align*}
 since $\calS_1(\ket{\bot})=\sum_{y'\in \domain}\frac{1}{\sqrt{d}} \ket{y'}$.
 Now
\begin{displaymath}
\calS_1^{\otimes (n-r)}\sum_{I\in \calB,\ y\in \domain^I} \eta_{I,y} \ket{y}_I \ket{\perp}_{[n] \setminus (I\cup V)}=
\sum_{z\in (\domain\cup\{\bot\})^{[n]\setminus V}} \delta_z \ket{z}_{n\setminus V}
 \end{displaymath}
 for some amplitudes $\delta_z$ satisfying $\sum_{z\in (\domain\cup \{\bot\})^{[n]\setminus V}}|\delta_z|^2=1$.
 
 For each value of $z\in \domain^{[n]\setminus V}$, since the
 sub-matrix $A_{U,V}$ is invertible, there is a unique value $y_z\in \domain^V$ such
 that $A_U (y_z \cup z)=q$ so
 we get that
\begin{align*}
&\norm[\bigg]{\Pi_q\calS_1^{\otimes n}\sum_{I\in \calB,\ y\in \domain^I} \eta_{I,y} \ket{y}_I \ket{\perp}_{[n] \setminus I}}^2 \\
    &=\norm[\bigg]{\Pi_q\bigg[\sum_{y'\in \domain^V}\frac{1}{\sqrt{d^r}}\ket{y'}_V\otimes\sum_{z\in (\domain\cup\{\bot\})^{n-r}} \delta_z \ket{z}_{[n]\setminus V}\bigg]}^2\\
     &=\norm[\bigg]{\frac{1}{\sqrt{d^r}}\cdot \Pi_q\bigg[\sum_{y'\in \domain^V}\ket{y'}_V\sum_{z\in \domain^{n-r}} \delta_z \ket{z}_{n\setminus V}\bigg]}^2\\
   & =\norm[\bigg]{\frac{1}{\sqrt{d^r}}\cdot \Pi_q\sum_{z\in \domain^{[n]\setminus V}}\delta_z\sum_{y'\in \domain^V} \ket{y'}_V \ket{z}_{n\setminus V})}^2
   \\
   & =\norm[\bigg]{\frac{1}{\sqrt{d^r}}\sum_{z\in \domain^{[n]\setminus V}}\delta_z\ket{y_z}_V \ket{z}_{n\setminus V})}^2\\
   &\le \frac{1}{d^r}
\end{align*}
since $\sum_{z\in \domain^{[n]\setminus V}}|\delta_z|^2\le 1$.
\end{proof}

Next we decompose the set of all $I$ with
$|I|\le t$ into buckets where we can apply the
above with $r$ equal to a constant fraction
of $k$.

\begin{lemma}
\label{lem:partition}
Let $A$ be a $(k,h,c)$-rigid matrix and let $k'=\lceil ck\rceil$.  Then for every subset $U$ of $k$
rows of $A$, there is a collection of disjoint $k'$-subsets of columns from $[n]$, $V_1,\ldots, V_\ell$ for $\ell=\lceil h/k'\rceil\le \lceil h/(ck)\rceil$ and corresponding sets of rows $U_1,\ldots, U_\ell\subseteq U$ such that for each $j\in [\ell]$, the $k'\times k'$ submatrix $A_{U_j,V_j}$ is full rank.  
(In particular the union, $W$, of the sets $V_j$ has size at least $h$.)
If $c=1$ then all $U_j=U$.
\end{lemma}

\begin{proof}
Fix $U\in [m]$ with $|U|=k$.  The following procedure constructs such a collection, one set at a time.   We maintain a subset of $W$ columns that is the union of the $V_j$ constructed so far.  
Suppose that $|W|<h$.   Then, by the $(k,h,c)$-rigidity of $A$, the
submatrix $A_{U,[n]\setminus W}$ has rank at least $k'$.   Hence there is a $k'\times k'$ submatrix $A_{U_j,V_j}$ of $A_{U,[n]\setminus W}$ that
has full rank $k'$.   
We now add $V_j$ to the collection of $k'$-sets of columns, record its corresponding row set
$U_j$, and set $W\leftarrow W\cup V_j$.
This produces exactly $\lceil h/k'\rceil$ subsets.
\end{proof}

Fix the collection of sets $V_1,\ldots,V_\ell$ 
given by \cref{lem:partition}.  
Let $k''=\lfloor \alpha k'\rfloor$.
Suppose that $V_j=\{i_1,\ldots, i_{k'}\}\subseteq [n]$ with 
$i_1\le \cdots\le i_{k'}$.
For each $\lambda\in \binom{[k']}{k''}$, define the
set $V^\lambda_j$ to be the subset of $V_j$ that has
the $k''$ elements of $V_j$ indexed by $\lambda$
removed.   (That is, $i_{j'}\notin V^\lambda_j$ iff
$j'\in \lambda$.)
Then $|V_j^\lambda|=k'-k''\ge c(1-\alpha)k$.
There are a total of 
$\binom{k'}{k''}\le 2^{H_2(\alpha)\, k'}$ possible values
of $\lambda$ and hence $\lceil h/k'\rceil \cdot 
2^{H_2(\alpha)\,k'}$ sets of the form $V_j^\lambda$.
These sets have two useful properties: first any subset of $[n]$ with size at most $\alpha h$ must miss some $V_j^\lambda$ and second if the entries of $x$ corresponding to some $V_j^\lambda$ are uniformly random, then for any set of $k$ indices in $Ax$, at least $c(1-\alpha)k$ of these values are also uniformly random.

\begin{lemma}
\label{lem:map}
For $t\le\alpha h$ and every $I\subseteq [n]$ with
$|I|\le t$, there is some $j\le \lceil h/k'\rceil$
and $\lambda\in \binom{[k']}{k''}$ such that 
$I\subseteq [n]\setminus V^\lambda_j$.
\end{lemma}

\begin{proof}
Fix such a set $I$ with $|I|\le t$.  
Since $t\le \alpha h$, $|\bigcup_{j\in [\ell]}V_j|\ge h$,  and the sets $V_j$ are disjoint, 
by averaging
there is some set $V_j$ that has at most an $\alpha$
fraction of its elements in $I$.
Hence  $V_j$ has at most $k''\le \alpha k'$ elements of $I$. 
Choose a set $\lambda\in \binom{[k']}{k''}$ that contains
the indices within $V_j$ of all of the elements of $V_j\cap I$.
Then by construction $I\cap V^\lambda_j=\varnothing$.
\end{proof}

By applying \cref{lem:map} we can associate each $I\subseteq [n]$ with $|I|\le t$ with
a pair $(j,\lambda)$ such that
$I\in [n]\setminus V^\lambda_j$ and define
bucket $\calB_j^\lambda$ to consist of all such sets $I$
associated with pair $(j,\lambda)$.
Further, define a set $U_j^\lambda\subseteq U_j\subseteq [m]$ of the rows of $A_{q(w)}$ with $|U_j^\lambda|=k'-k''$ such that the submatrix
$A_{U_j^\lambda,V_j^\lambda}$ is full rank.
Such a subset of rows must exist since $A_{U_j,V_j^\lambda}$ is a full rank matrix.
Then let $q_j^\lambda=q(w)|_{U_j^\lambda}$ be the
portion of the assignment $q(w)$ on the rows of
$U_j^\lambda$.

We are now ready to provide an upper bound on the success probability from \cref{eq:target}.
\begin{align}
&\norm[\bigg]{\Pi_{q(w)} \calS_1^{\otimes n} \sum_{\substack{I \subseteq [n],\ |I| \leq t \\ y \in \domain^I}} \beta^{i,p,w}_{I, y} \ket{y}_I \ket{\perp}_{[n] \setminus I}}^2\nonumber\\
&=\norm[\bigg]{\Pi_{q(w)} \calS_1^{\otimes n} \sum_{j\in [\ell]}\sum_{\lambda\in \binom{[k']}{k''}}\sum_{I \in \calB_j^\lambda,\ y \in \domain^I}\beta^{i,p,w}_{I, y} \ket{y}_I \ket{\perp}_{[n] \setminus I}}^2\nonumber\\
&\le\norm[\bigg]{ \sum_{j\in [\ell]}\sum_{\lambda\in \binom{[k']}{k''}}\Pi_{q_j^\lambda}\ \calS_1^{\otimes n}\sum_{I \in \calB_j^\lambda,\ y \in \domain^I}\beta^{i,p,w}_{I, y} \ket{y}_I \ket{\perp}_{[n] \setminus I}}^2.\label{eq:bound}
\end{align}
Applying \cref{lem:group} with $r=k'-k''$, $q=q_j^\lambda$,  $U=U_j^\lambda$, $V=V_j^\lambda$,
and $\calB=\calB_j^\lambda$, we have that
\begin{displaymath}\norm[\bigg]{\Pi_{q_j^\lambda}\ \calS_1^{\otimes n}\sum_{I \in \calB_j^\lambda,\ y \in \domain^I}\beta^{i,p,w}_{I, y} \ket{y}_I \ket{\perp}_{[n] \setminus I}}^2\le 1/d^{k'-k''}\le 1/d^{(1-\alpha)\,k'}.
\end{displaymath}
and hence using \cref{eq:bound} we obtain that

\begin{displaymath}\norm[\big]{ \Pi_{k} \calS \ket{\phi_t}}^2
\le
\ell^2 \ \binom{k'}{k''}^2/d^{(1-\alpha)\,k'}\le
\lceil h/k'\rceil^2 \ 4^{  H_2(\alpha)\, k'}/d^{(1-\alpha)\,k'}
= \lceil h/k'\rceil\  (4^{H_2(\alpha)}/d^{(1-\alpha)})^{k'}.
\end{displaymath}
Without loss of generality in our desired bound we can assume that $4^{H_2(\alpha)}/d^{(1-\alpha)}<1$.  Therefore the bound still applies when we replace $k'$
by the potentially smaller $ck$ which is what we needed to show.
\end{proof}

\subsection{Related time-space tradeoff and cumulative memory lower bounds}
Following the same arguments as for classical computation \cite{Abr91}, we use \cref{thm:time-space-matrix-vector} to obtain a collection of time-space lower bounds for problems that are closely related to matrix vector products. Our proofs are identical to their classical counterparts proven in\cite[Sections 5-6]{Abr91} and are duplicated here for completeness.
Many of these lower bounds are tightly matched by classical query algorithms.
Constructions of matching upper bounds can be found in \cref{sec:query-algs}.

\begin{corollary}\label{cor:DFT-ts-lb}
    Let $\F$ be a field and $\domain \subseteq \F$ such that $d=|\domain|$. Any quantum circuit that computes the discrete Fourier transform (DFT) of vectors in $\domain^n$ in time $T$ and space $S$ with probability at least $2^{-S}$ requires $T$ to be $\Omega(n^2 \log (d) \ / S)$.
\end{corollary}
\begin{proof}
    Applying \cref{thm:time-space-matrix-vector} with the rigidity of the DFT from \cref{prop:DFT-rigid} directly gives us the lower bound.
\end{proof}
\begin{proposition}[\cite{Abr91}]\label{prop:toeplitz-rigid}
There is a constant $\gamma \in (0,1/2)$ such that at least a $1-|\domain|^{-1}(2/3)^{\gamma n}$ fraction of the Toeplitz (diagonal constant) matrices over $\domain^{n \times n}$ are $(\gamma n, \gamma n)$-rigid.
\end{proposition}
Recall that the convolution of two vectors $w = u*v$ is $w_k = \sum_{i \in [n]} u_i v_{k-i}$ where the indices are reduced modulo $n$, where we identify $n$ with 0.
\begin{corollary}\label{cor:convo-ts-lb}
    Let $\F$ be a field and $\domain \subseteq \F$ such that $d=|\domain|$.  Any quantum query algorithm computing the convolution of two vectors in $\domain^n$ in time $T$ and space $S$ with probability at least $2^{-S}$ requires $T$ to be $\Omega(n^2 \log (d) \ / S)$
\end{corollary}
\begin{proof}
    For simplicity assume that $n$ is even. Let
    \begin{equation*}
        U = \begin{bmatrix} u_{n}& u_{n-1}&\hdots& u_2& u_1\\
        u_{1}& u_n & \hdots& u_3& u_2 \\ \vdots & \vdots & \ddots & \vdots & \vdots\\ u_{n-2} & u_{n-3} & \hdots & u_{n} & u_{n-1} \\ u_{n-1} & u_{n-2} & \hdots & u_1 & u_n\end{bmatrix} = \begin{bmatrix} A & B \\ C & D\end{bmatrix}
    \end{equation*}
    Where $A,B,C \text{ and } D$ are $n/2 \times n/2$ submatrices. Then $Uv$ is the convolution between vectors $u$ and $v$.
    Observe that $U$ is a Toeplitz matrix and by picking $u$ to be a uniform vector over $\domain$, 
    \cref{prop:toeplitz-rigid} tells us that for sufficiently large $n$, there is a constant $\gamma \in (0,1/2)$ such that both $A$ and $B$ are $(\gamma n, \gamma n/2)$-rigid with probability at least $1/2$.
    This lets us restrict our input to such choices for $u$ and observe that the matrix $U' = \begin{bmatrix} A & B \end{bmatrix}$ is $(\gamma n, \gamma n/2)$-rigid, so \cref{thm:time-space-matrix-vector} gives us that computing $U'v$ requires $T$ that is $\Omega(n^2 \log (d) \ / S)$. Since $U'$ is a subfunction of $U$, convolution also requires $T$ that is $\Omega(n^2 \log (d)\ / S)$.
\end{proof}
\begin{corollary}\label{cor:bin-mult-ts-lb}
    A quantum circuit that multiplies two $n$ bit binary numbers in time $T$ and space $S$ with probability at least $2^{-S}$ requires $T$ to be $\Omega(n^2 / (S \log^2 n))$.
\end{corollary}
\begin{proof}
    Let $u,v$ be arbitrary vectors over $\F_2$. Define the binary number
    \begin{equation*}
        u'=0^{\ceil{\log_2 n} -1} u_n \ldots 0^{\ceil{\log_2 n} -1} u_1 0^{\ceil{\log_2 n} -1} u_n \ldots 0^{\ceil{\log_2 n} -1} u_1
    \end{equation*} and similarly define $v'$. Then observe that the product $u' \cdot v'$ contains all entries of the convolution between $u$ and $v$ encoded in blocks of $\ceil{\log_2 n}$ bits each. By \cref{cor:convo-ts-lb} this requires $T$ to be $\Omega(n^2 / (S\log^2 n))$.
\end{proof}
\begin{proposition}[\cite{Abr91}]\label{prop:mat-mat-mat-as-mat-vec}
        Let $A,B,C,Y \in \domain^{n \times n}$.
        Let $\calB$ (and $\calY$) be the vectors in $\domain^{n^2}$ formed by stacking the transposes of the rows of $B$ (and $Y$) into a column vector.
    If $\domain$ is a commutative ring, then the following conditions are equivalent:
    \begin{align*}
        Y &= ABC\\
        \calY &= (A \otimes C^T) \calB
    \end{align*}
    Where $\otimes$ is the standard tensor (Kronecker) product.
\end{proposition}
\begin{sloppypar}
\begin{proposition}[\cite{Abr91}]\label{prop:tensor-rigid-is-rigid}
    Let $\gamma \in (0,1/2)$. If $A$ and $B$ are $(\gamma n, \gamma n)$-rigid, then $A \otimes B$ is $(\gamma^2 n^2, \gamma^2 n^2, \gamma^2)$-rigid.
\end{proposition}
\end{sloppypar}
\begin{corollary}\label{cor:mat-mat-mat-product-ts-lb}
    Let $\F$ be a field and $\domain \subseteq \F$ such that $d=|\domain|$. Any quantum circuit that computes the product $ABC$ on inputs $A,B,C \in \domain^{n \times n}$ in time $T$ and space $S$ with probability at least $2^{-S}$ requires $T$ that is $\Omega(n^4 \log (d)\ / S)$.
\end{corollary}
\begin{proof}
    We use \cref{prop:mat-mat-mat-as-mat-vec} to view this as a matrix-vector product problem where $\calB$ is the input and $\calY$ is the output.
    By \cref{prop:rigid-matrices} there is a constant $\gamma \in (0,1/2)$ such that both $A$ and $C$ are $\gamma$ rigid with constant probability, so we can assume such without increasing the expected cost by more than a constant factor.
    Then \cref{prop:tensor-rigid-is-rigid} gives us that $A \otimes C$ is $(\gamma^2 n^2, \gamma^2 n^2, \gamma^2)$-rigid and we can apply \cref{thm:time-space-matrix-vector} to get that $T$ must be $\Omega(n^4 \log (d) \ / S)$ as desired.
\end{proof}

\begin{corollary}\label{cor:mat-cube-ts-lb}
    Let $\F$ be a field and $\domain \subseteq \F$ such that $d=|\domain|$. Any quantum circuit that computes $A^3$ on inputs in $\domain^{n \times n}$ in time $T$ and space $S$ with probability at least $2^{-S}$ requires $T$ that is $\Omega(n^4 \log (d) \ / S)$.
\end{corollary}
\begin{proof}
    Let $A,B,C \in \domain^{n \times n}$. Then construct the $4n \times 4n$ matrix:
    \begin{equation*}
        M = \begin{bmatrix}
            0 & A & 0 & 0\\
            0 & 0 & B & 0\\
            0 & 0 & 0 & C\\
            0 & 0 & 0 & 0
        \end{bmatrix}
    \end{equation*}
    Observe that the top right $n \times n$ sub-matrix of $M^3$ is equal to the product $ABC$. Thus we get a reduction to matrix-matrix-matrix product and can apply \cref{cor:mat-mat-mat-product-ts-lb} to get our lower bound.
\end{proof}

\begin{corollary}\label{cor:mat-invert-ts-lb}
    Let $\F$ be a field and $\domain \subseteq \F$ such that $d=|\domain|$. Any quantum circuit that computes $A^{-1}$ on unit upper triangular inputs in $\domain^{n \times n}$ in time $T$ and space $S$ with probability at least $2^{-S}$ requires $T$ that is $\Omega(n^4 \log(d) / S)$.
\end{corollary}
\begin{proof}
    Let $A,B,C \in \domain^{n \times n}$. Then construct the $4n \times 4n$ matrix:
    \begin{equation*}
        M = \begin{bmatrix}
            I & -A & 0 & 0\\
            0 & I & -B & 0\\
            0 & 0 & I & -C\\
            0 & 0 & 0 & I
        \end{bmatrix}
    \end{equation*}
    Where $I$ is the $n \times n$ identity submatrix.
    Then observe that $M^{-1}$ has the product $ABC$ as its top right $n \times n$ submatrix. We can again use \cref{thm:time-space-matrix-vector} to get our lower bound.
\end{proof}

\begin{corollary}\label{cor:system-eqn-cm-lb}
    Let $\F$ be a field and $\domain \subseteq \F$ such that $d=|\domain|$. Any quantum circuit that solves any $n \times n$ system of linear equations over $\domain$ in time $T$ and space $S$ with probability at least $2^{-S}$ requires $T$ that is $\Omega(n^3 \log (d) \ / S)$
\end{corollary}
\begin{proof}
    It is possible to invert a matrix by solving $n$ systems of $n$ linear equations. By a reduction \cref{cor:mat-invert-ts-lb} gives us that solving these equations requires $T$ that is $\Omega(n^4 \log(d) \ / S)$. Thus least one of these equations must require $T$ that is $\Omega(n^3 \log (d) \ / S)$ to solve.
\end{proof}

In \cite{BK23} the authors showed that the kinds of quantum time-space product lower bounds we proved in this section can be extended to asymptotically equivalent lower bounds on the stronger notion of cumulative memory complexity. We restate a simplified version of their main theorem for quantum circuits and classical query algorithms here.
\begin{proposition}[\cite{BK23}]\label{prop:TS-to-cm-lb}
    Let $f: \domain^n \to \range^m$ be a function such that there exists constant $C$, functions $m'(n) \in \omega(\log n), h(k,n) = k^\Delta h_1(n), K(n)$, and a distribution $\mu$ over $\domain^n$ where when $x \sim \mu$ the probability that - for any $k \leq m'(n)$ - any quantum circuit (or classical query algorithm) with at most $h(k,n)$ queries to $x$ produces $k$ correct output values of $f(x)$ with probability at most $C \cdot K(n)^{-k}$.
    Then for any constant $c>0$, any quantum circuit (or classical query algorithm) that computes $f$ with $T$ queries and error $\epsilon \leq (1-1/(2T^c))$ must have cumulative memory that is:
    \begin{equation}
        \Omega \left(\min\left([(m h_1(n))^{1/(1-\Delta)} \log K(n)] / T^{\Delta/(1-\Delta)}, m'(n)^{1+\Delta}h_1(n) \log K(n) \right) \right)
    \end{equation}
\end{proposition}
Using the above result, we can extend the quantum time-space product lower bound for matrix vector products to a matching quantum cumulative memory lower bound.
\begin{theorem}\label{cor:mat-vec-cm-lb}
    Let $\gamma>0$ and $c \in (0,1/2]$ be fixed. If $A$ is a $(\gamma n, \gamma n, c)$-rigid $n\times n$ matrix over a field $\F$ then any quantum circuit using time $T$ and space $S$ that computes the function $f_A: \domain^n \to \F^n$ for $\domain\subseteq \F$ with $d=|\domain|$ given by $f_A(x) = Ax$ with success probability larger than $1/T$ requires cumulative memory that is $\Omega(n^2 \log d)$.
\end{theorem}
\begin{proof}
    By \cref{lem:matvec} we can apply \cref{prop:TS-to-cm-lb} where $C = \ceil{1/c}$, $m'(n) = \gamma n$, $\Delta = 0$, $h_1(n) = \alpha n$, $K(n) = d^{1/6}$, and $\mu$ is the uniform distribution.
    This give us that any quantum circuit computing $f_A$ with $T$ queries and error at most $1-1/(2T)$ requires cumulative memory $\Omega(n^2 \log d)$ as desired.
\end{proof}
Directly applying this in place of \cref{thm:quantum-booolean-matrix} gives us matching cumulative $(CM)$ memory lower bounds for \cref{cor:DFT-ts-lb} through \cref{cor:system-eqn-cm-lb}.
\begin{corollary}
    Let $\F$ be a field and $\domain \subseteq \F$ such that $d=|\domain|$.
    Any quantum circuit with inputs over $\domain$ that computes the DFT or vector convolution requires $CM$ that is $\Omega(n^2 \log d)$.
    Any quantum circuit that computes the product of three matrices, matrix cubing, or matrix inversion requires $CM$ that is $\Omega(n^4 \log d)$.
    Any quantum circuit that solves $n \times n$ systems of linear equations requires $CM$ that is $\Omega(n^3 \log d)$.
    Additionally any quantum circuit that multiplies two $n$ bit binary numbers requires $CM$ that is $\Omega(n^2 / \log^2 n)$.
\end{corollary}
\section{Quantum matrix multiplication}\label{sec:mat-mult}

While many of the applications so far, including the matrix
triple product lower bound discussed in the
previous section, are derived from the matrix-vector
product lower bound, 
our matrix multiplication
lower bound requires a separate argument using ideas from the classical lower bound for the problem in~\cite{Abr91}.
Implementing this requires a
much more subtle way of applying
our bucketing method for states that allows us to concentrate
on just a subset of the buckets containing most of the total amplitude and ignore the 
others.
As in \cref{sec:mat-vec}, our lower bounds in this section apply to a more general model of quantum circuits that can decide which outputs they want to produce in a given layer based on the inputs that they have queried.

Here we consider the matrix multiplication problem $f(A,B) = AB$ where both $A$ and $B$ are considered input.
If we could fix a choice of $A$, we would be able to make our proof somewhat simpler.
However, as Abrahamson pointed out in \cite{Abr91}, there is a classical algorithm that can compute the function $f(B) = AB$ for any fixed matrix $A$ in $O(n^2)$ queries and $O(n \log d)$ space.
Thus our lower bound requires both $A$ and $B$ to be inputs to the function.

\begin{theorem}\label{thm:mat-mult}
Let $\F$ be a field and $\domain \subseteq \F$ with $d = |\domain|$. Then any quantum circuit $\calC$ that uses time $T$ and space $S$ and computes the function $f: \domain^{2n^2} \to \F^{n^2}$ given by $f(A,B) = AB$ with success probability larger than $1/T$ must have $T$ that is $\Omega(n^3 \sqrt{\log d\ /S})$.
\end{theorem}

Again this theorem follows from the following key lemma, proven in \cref{subsec:quant-mat-mat-lemma}, which lets us bound the number of correct output values produced by a shallow quantum circuit.

\begin{lemma}\label{lem:mat-mul-tail}
Let $\gamma\in (0,1/2)$ and
$f: \domain^{n^2} \times \domain^{n^2} \to \F^{n^2}$ for $\domain \subseteq \F$ with $|\domain|=d$ be defined by $f(A,B) = AB$. Then for any constant $\beta>0$ and quantum circuit $\calC$ with at most $h =  \beta \gamma n \sqrt{k/2}$ queries to input matrices $A,B$ sampled uniformly from $\domain^{n^2}$, the probability that $A$ and $B$ are $(\gamma n, \gamma n)$-rigid and $\calC$ produces $k$ correct output values of $f(A,B)$ is at most $16 \min(k,n)^{\sqrt{2k}} (4^{H_2(4 \beta)}/d^{1-4\beta})^{k/4}$
\end{lemma}

\begin{sloppypar}
Note that for $\beta \leq 0.0184$ we have $1-4\beta -2H_2(4\beta)>1/6$ so the bound is at most $16\min(k,n)^{\sqrt{2k}} d^{-k/24}$.
\end{sloppypar}

\begin{proof}[Proof of \cref{thm:mat-mult} from \cref{lem:mat-mul-tail}]
Let $\gamma\in (0,1/2)$ be the constant given by
\cref{prop:rigid-matrices}.
By that proposition, the probability
that either of two matrices $A$ and $B$ chosen uniformly randomly from
$\domain^{n^2}$ is not $(\gamma n,\gamma n)$-rigid is
at most $2 d^{-1} (2/3)^{\gamma n}$.
    Let $\calC$ be a quantum circuit with $T$ queries and space $S$.
    Let $\beta = 0.0429$, $d = |\domain|$, and set 
    \begin{math}
        k = \ceil{48(6S+ 4)/\log_2 d}.
    \end{math}
    We partition $\calC$ into $\ceil{T/(\beta \gamma n \sqrt{k/2})}$ sub-circuits that each have at most $\beta \gamma n\sqrt{k/2}$ queries.  Without loss of generalities there are at most 
    $n^2$ such sub-circuits.
    By combining \cref{prop:quant-union} with \cref{lem:mat-mul-tail}, we know that for a uniformly random input, the probability that $A$ and $B$ are $(\gamma n, \gamma n)$-rigid matrices and a fixed sub-circuit can produce $k$ outputs is at most $16 \min(k,n)^{\sqrt{2k}} 2^{S} d^{-k/24} \leq 16 k^{\sqrt{2k}} 2^{S} d^{-k/24}$.
    Therefore the probability that $A$ and $B$ are
    $(\gamma n, \gamma n)$-rigid matrices and one of
    the
    sub-circuits produces $k$ correct output values
    is at most $16 k^{\sqrt{2k}} 2^{S} d^{-k/24} n^2$.
    Combining this with the probability that one of
    $A$ or $B$ is not $(\gamma n,\gamma n)$-rigid, the
    probability that there is a 
    sub-circuit that correctly produces $k$ 
    output values is at most 
    \begin{displaymath}
      16 k^{\sqrt{2k}} 2^{S} d^{-k/24}n^2 + 2 d^{-1}(2/3)^{2\gamma n}.  
    \end{displaymath}
    Since we can assume without loss of generality that $T\le n^3$, for sufficiently large $n$, $2 d^{-1}(2/3)^{2\gamma n} \leq 1/(2T)$ and $k^{\sqrt{2k}} \leq 2^{k/48} \leq d^{k/48}$.
    Plugging in our value of $k$ and the fact
    that $S\ge \log_2 n$ without loss of generality gives a probability of
    at most 
    \begin{align*}
        16 k^{\sqrt{2k}} 2^{2S} d^{-k/24} n^2 + 2d^{-1} (2/3)^{2\gamma n} & \leq 16 2^{S} d^{-k/48} n^2 + 1/(2T)\\
        &\le 1/(2T) + 1/(2T) = 1/T.
    \end{align*}
    Since $\calC$ must be correct with probability larger than $1/T$, this implies that
    \begin{equation*}
        (k-1) \ceil{T/(\beta \gamma n \sqrt{k/2})} \geq n^2.
    \end{equation*}
    Plugging in our value of $k$ gives us that
    \begin{equation*}
        T \text{ is } \Omega(n^3 \sqrt{\log d}/\sqrt{S + \log T}).
    \end{equation*}
    Since $S \geq \log_2 n$ and our bound trivially holds when $T$ is $\omega(n^3 \sqrt{\log d})$ there is a constant $c>0$ such that $cS \geq \log_2 T$.
    This implies that $T$ is $\Omega(n^3 \sqrt{\log d / S})$ as desired.
\end{proof}
Our quantum lower bound is tightly matched by a classical query algorithm in \cref{prop:matrix-product-upperbound}.

\subsection{The success probability of small depth quantum circuits}\label{subsec:quant-mat-mat-lemma}

We first give an overview of the argument which assumes a uniform 
distribution over all input matrices
$A$ and $B$ in $D^{n\times n}$.
Unlike the matrix-vector product
proof, in addition to the requirement of
$k$ correct output values, for success we also include
the extra condition that both matrices
must be $(\gamma n,\gamma n)$ rigid.
As in the case of the matrix-vector product proof, we decompose the state after $t\le h=\beta \gamma n\sqrt{k/2}$ steps into orthogonal components based on different values $\ket{i,p,w}$ which determines
the $k$ output values produced, though this now can be up to
quadratic in $n$. 
However, unlike that proof, we need to use the weighted version of our bucketing method.
It again suffices to show that for each such $\ket{i,p,w}$
the total fraction of the squared amplitude for any state that is spanned by recording query basis states with at most $t$ non-$\bot$ items can correspond to inputs where there is success is 
exponentially small in $k$.  

The output values produced determine a set of rows of the matrix $A$ and columns of the matrix $B$ that are relevant.    
For classical algorithms, where we can determine a set of input locations queried, the lower bound of~\cite{Abr91} shows that either at least $k/4$ of the output values lie in rows where few elements
of $A$ are queried or $k/4$ lie in columns where
few elements of $B$ are queried.
For each of these cases (``light'' rows or ``light'' columns). 
The corresponding output values in those rows or columns are hard to produce in that the requirement that the other matrix is rigid means that the algorithm is exponentially unlikely in $k$ to be correct on those entries.

In the quantum case, when viewed in the recording query basis, the state involves a superposition over all possible assignments to
subsets of indices for the relevant rows of $A$ and columns of $B$ with at most $t$ non-$\bot$
entries.
For convenience, we first split these basis states
depending on whether there are many 
outputs in light rows or many in light columns; and then on which rows/columns those are; each determines a set of $k/4$ output values to consider hard and whether to focus on matrix $A$ or $B$. The number of such possibilities is not too large so the total is not too much
larger than the maximum over all such
choices.  
We further consider a fixed choice of the other rigid matrix that maximizes
the resulting probability that the hard outputs produced have correct values.
The number of consistent recording query basis states in each such superposition is still 
enormous. 

We need to apply bucketing where either $A$ or $B$ is fixed as a rigid matrix and the other can be interpreted as a having a collection of light columns (or rows) such that the output values are the results of a matrix-vector products
involving vectors with few queries.   
However, repeatedly applying the basic bucketing method for basis states we used for matrix-vector
products fails because the total number of buckets would be too large since it would end up being the product over the number of choices for each row or column.

Instead, we show that among these potential buckets we can find a small number of them that together captures a large portion of the
amplitude associated with the state, letting us derive the final lower bound.
We now give the details of this argument.

\begin{proof}[Proof of \cref{lem:mat-mul-tail}]
Let $C = AB$, $\Pi_{\text{rigid}(A)}$ (and $\Pi_{
\text{rigid}(B)}$) be the projection onto inputs where $A$ (and $B$) are $(\gamma n, \gamma n)$-rigid matrices, and define $\Pi_{\text{rigid}} = \Pi_{\text{rigid } A} \Pi_{
\text{rigid } B}$.
Assume that $q(w)$ --- the output as a function of the measured value of the work register --- produces exactly $k$ outputs; we ignore anything it produces
after the first $k$.
We will use $[A]$ to denote the set of indices of
elements in $A$ and likewise for $[B]$ and $[C]$.
By \cref{prop:max-queried}, after $t \leq h$ queries in the recording query basis, 
 the state $\ket{\phi_t}$
 is a linear combination of basis states $\ket{i, p, w, x_1, \ldots, x_n}$ where $(x_1, \ldots, x_n) \in \Gamma_t$.
As in our analysis of the case of matrix-vector products, it will be necessary to be more 
 explicit in our discussion of $\Gamma_t$. 
 Each element of $\Gamma_t$ consists
 of an assignment $x\in \domain^E$ 
 and $y\in \domain^F$
 for some subsets $E\subseteq [A]$ 
 and $F\subseteq [B]$ with
 $|E|+|F|\le t$ and value $\perp$ on all
 coordinates in $[A]\setminus E$ and
 $[B]\setminus F$.
 Therefore,
our state can be written as:
\begin{equation*}
    \ket{\phi_t} = \sum_{\substack{i,p,w\\E \subseteq [A], F \subseteq [B]\\ |E| + |F| \leq t \\ x \in \domain^E, y \in \domain^F}} \alpha_{i,p,w,E,F,x,y} \ket{i,p,w}\ket{x}_E\ket{\perp}_{[A] \setminus E} \ket{y}_F \ket{\perp}_{[B] \setminus F}
\end{equation*}
for some $\alpha_{i,p,w,E,F,x,y}$ with $\sum_{i,p,w,E,F,x,y} |\alpha_{i,p,w,E,F,x,y}|^2 = 1$.  
We first apply an analogous series of observations and decompositions to those that allowed us to derive (\ref{eq:target}) from
(\ref{eq:matrix-vector-1}) in the case of matrix-vector product.
By \cref{prop:recording-equiv}, we note that the final state of the algorithm in the standard oracle setting is given by:
\begin{equation*}
    \ket{\psi_t} = \calS \ket{\phi_t} = \calS \sum_{\substack{i,p,w\\E \subseteq [A], F \subseteq [B]\\ |E| + |F| \leq t \\ x \in \domain^E, y \in \domain^F}} \alpha_{i,p,w,E,F,x,y} \ket{i,p,w}\ket{x}_E\ket{\perp}_{[A] \setminus E} \ket{y}_F \ket{\perp}_{[B] \setminus F}
\end{equation*}
Because $\calS$ behaves as the identity on $\ket{\phi_t}_\calC$ and each distinct choice of $\ket{i,p,w}$ gives an orthogonal basis state, this equals:
\begin{equation*}
    \sum_{i,p,w} \beta_{i,p,w} \ket{i,p,w} \otimes \bigg[ S_1^{\otimes 2n^2} \sum_{\substack{E \subseteq [A], F \subseteq [B]\\ |E|+|F| \leq t \\ x \in \domain^E, y \in \domain^F}} \beta^{i,p,w}_{E,F,x,y} \ket{x}_{E}\ket{\perp}_{[A] \setminus E}\ket{y}_{F}\ket{\perp}_{[B] \setminus F} \bigg]
\end{equation*}
for some $\beta_{i,p,w}$ and $\beta^{i,p,w}_{E,F,x,y}$ such that $\sum_{i,p,w} |\beta_{i,p,w}|^2 = 1$ and   $\sum_{E,F,x,y} |\beta^{i,p,w}_{E,F,x,y}|^2 = 1$ for each $i,p,w$.
Now the probability over the choices of the input matrices and the result of the quantum
algorithm making $t$ queries that the matrices $A$ and $B$ are both $(\gamma n, \gamma n)$-rigid 
and the algorithm produces $k$ correct output values from $C=AB$ is at most:
\allowdisplaybreaks[1]
\begin{align}
  &\norm[\big]{\Pi_k \Pi_{\text{rigid}} \calS \ket{\phi_t}}^2 \nonumber\\
    &= \norm[\bigg]{\Pi_k \Pi_{\text{rigid}} \sum_{i,p,w} \beta_{i,p,w} \ket{i,p,w} \otimes \bigg[ S_1^{\otimes 2n^2} \sum_{\substack{E \subseteq [A], F \subseteq [B]\\ |E|+|F| \leq t \\ x \in \domain^E, y \in \domain^F}} \beta^{i,p,w}_{E,F,x,y} \ket{x}_{E}\ket{\perp}_{[A] \setminus E}\ket{y}_{F}\ket{\perp}_{[B] \setminus F} \bigg]}^2\nonumber\\
    &= \norm[\bigg]{\sum_{i,p,w} \beta_{i,p,w} \ket{i,p,w} \otimes \bigg[\Pi_{q(w)} \Pi_{\text{rigid}} S_1^{\otimes 2n^2} \sum_{\substack{E \subseteq [A], F \subseteq [B]\\ |E|+|F| \leq t \\ x \in \domain^E, y \in \domain^F}} \beta^{i,p,w}_{E,F,x,y} \ket{x}_{E}\ket{\perp}_{[A] \setminus E}\ket{y}_{F}\ket{\perp}_{[B] \setminus F} \bigg]}^2\nonumber\\
    &= \sum_{i,p,w} |\beta_{i,p,w}|^2 \norm[\bigg]{\bigg[\Pi_{q(w)} \Pi_{\text{rigid}} S_1^{\otimes 2n^2} \sum_{\substack{E \subseteq [A], F \subseteq [B]\\ |E|+|F| \leq t \\ x \in \domain^E, y \in \domain^F}} \beta^{i,p,w}_{E,F,x,y} \ket{x}_{E}\ket{\perp}_{[A] \setminus E}\ket{y}_{F}\ket{\perp}_{[B] \setminus F} \bigg]}^2\nonumber\\
    &\leq \max_{i,p,w} \norm[\bigg]{\Pi_{q(w)} \Pi_{\text{rigid}}\calS_1^{\otimes 2n^2} \sum_{\substack{E \subseteq [A], F \subseteq [B] \\ |E| + |F| \leq t \\ x \in \domain^E, y \in \domain^F}}
    \beta^{i,p,w}_{E,F,x,y} \ket{x}_E \ket{\perp}_{[A] \setminus E} \ket{y}_F 
    \ket{\perp}_{[B] \setminus F}}^2 .   \label{eq:matrix-mult-target}
\end{align}
For the rest of the proof we fix an $i,p,w$ to 
achieve the maximum value in \cref{eq:matrix-mult-target} and prove an upper bound
on the resulting probability.
This fixes the output values $q(w)$;
we write $G\subseteq [C]$ with $|G|=k$ for the set of
indices of the outputs given by $q(w)$.
To keep notations simpler in the remainder of the proof we observe that
\cref{eq:matrix-mult-target} is upper
bounded by the maximum of 
\begin{equation}
    \norm[\bigg]{\Pi_{q(G)} \Pi_{\text{rigid}}\calS_1^{\otimes 2n^2} \sum_{\substack{E \subseteq [A], F \subseteq [B] \\ |E| , |F| \leq t \\ x \in \domain^E, y \in \domain^F}}
    \beta_{E,F,x,y} \ket{x}_E \ket{\perp}_{[A] \setminus E} \ket{y}_F 
    \ket{\perp}_{[B] \setminus F}}^2    \label{eq:matrix-mult-target-simple}
\end{equation}
over 
all $\beta_{E,F,x,y}$ with
$\sum_{E,F,x,y} |\beta_{E,F,x,y}|^2=1$, all sets $G\subseteq [C]$ with $|G|=k$ and all assignments $q(G)$ to $G$.

We will split the sum in \cref{eq:matrix-mult-target-simple} over the different
sets $E$ and $F$ of queried input indices
depending on how they relate to the set of
output indices given by $G$.
Let $r(G)$ be the set of rows containing elements of $G$ and
$c(G)$ be the set of columns containing elements of $G$.\footnote{We will think of $r(G)$ and $c(G)$ as being subsets of indices in $[n]$ that correspond to rows in $A$ and columns of $B$, respectively, that are relevant for the outputs in $G$.}

Recall our bound $h=\beta\gamma n\sqrt{k/2}$ on the number of
queries.  
We define a \emph{light row of $E$} to be an element
of $r(G)$ that 
contains at most $\beta\gamma n$ elements of $E$
and define a \emph{light column of $F$} to be an element
of $c(G)$ that contains at most $\beta\gamma n$ elements of $F$.
Since $|E|+|F|\le t\le \beta\gamma n\sqrt{k/2}$ we have
$\le \sqrt{k/2}$ rows of $E$ in $r(G)$ and $\le \sqrt{k/2}$ columns of $F$ in $c(G)$ that are not light.
We define
$\calL(E)\subseteq r(G)$, to be the set of
light rows of $E$ and
$\calL'(F)\subseteq c(G)$ to be the
set of light columns of $F$. 
Therefore $|\{(i',j')\in G\mid i'\notin \calL(E),\ j'\notin \calL'(F)\}|\le k/2$
so at least $k/2$ elements of $G$
are in light rows of $E$ or in light columns of $F$.
Therefore for every pair $(E,F)$ at least one of
the sets of outputs $G^r_{\calL(E)}=\{(i',j')\in G\mid i'\in \calL(E)\}$ or
$G^c_{\calL'(F)}=\{(i',j')\in G\mid j'\in \calL'(F)\}$ 
has size $\ge k/4$.

Let $\calE$ be the set of all $E\subseteq [A]$ with $|E|\le t$ such that $G$ has at least $k/4$ outputs in light rows and $\calF$ be the set of all
$F\subseteq [B]$ with $|F|\le t$ such that $G$ has at least $k/4$ outputs in light columns.
We separately bound the contribution to \cref{eq:matrix-mult-target-simple} from pairs $(E,F)$ with $E\in \calE$ or $F\in \calF$.
The analyses of the two cases are completely symmetric up to matrix transposition. 
It will be convenient to focus on the case $F\in \calF$ representing basis states where there are many outputs of $G$
in light columns and compute an upper bound on
\begin{equation}
    \norm[\bigg]{\Pi_{q(G)} \Pi_{\text{rigid}}\calS_1^{\otimes 2n^2} \sum_{\substack{E \subseteq [A]\\|E|\le t\\
x\in \domain^E}} \sum_{\substack{F \in \calF \\   y \in \domain^F}}
    \beta_{E,F,x,y} \ket{x}_E \ket{\perp}_{[A] \setminus E} \ket{y}_F 
    \ket{\perp}_{[B] \setminus F}}^2 .   \label{eq:matrix-mult-target-light-columns}
\end{equation}
Basis states where $E\in \calE$ give exactly the same upper bound as \cref{eq:matrix-mult-target-light-columns} by applying the argument
 to the transposed product $B^T A^T$ and corresponding transposed sets $F^T$, $E^T$, and $G^T$.
 Hence, the quantity in \cref{eq:matrix-mult-target-simple}
 is at most 4 times that of \cref{eq:matrix-mult-target-light-columns}.

\begin{sloppypar}
 To upper bound \cref{eq:matrix-mult-target-light-columns}, we 
first remove the projection operator $\Pi_{\text{rigid } B}$
from $\Pi_{q(G)} \Pi_{\text{rigid}}=\Pi_{q(G)} \Pi_{\text{rigid } A} \Pi_{\text{rigid } B}$
to get
$\Pi_{q(G)} \Pi_{\text{rigid } A}$.
We then rewrite this combined projection operator as
$\Pi_{q(G)} \Pi_{\text{rigid } A}
=\sum_{A\ (\gamma n, \gamma n)\textbf{-rigid}}\Pi_A\otimes \Pi^A_{q(G)} $
where $\Pi_A$ is the projection onto the specific matrix
$A$ and for each $A$, $\Pi^A_{q(G)}$ is the projection onto the choices for
matrix $B$ such that $C=AB$ agrees with $q(w)$.
We therefore obtain that \cref{eq:matrix-mult-target-light-columns} is at most
\begin{align}
    &\norm[\bigg]{\sum_{A\ (\gamma n,\gamma n)\text{-rigid}} (\Pi_A\otimes\Pi^A_{q(G)} )\calS_1^{\otimes 2n^2} \sum_{\substack{E \subseteq [A]\\|E|\le t\\
x\in \domain^E}} \sum_{\substack{F \in \calF\\  y \in \domain^F}}
    \beta_{E,F,x,y} \ket{x}_E \ket{\perp}_{[A] \setminus E} \ket{y}_F 
    \ket{\perp}_{[B] \setminus F}}^2\nonumber\\
    &=\norm[\bigg]{\sum_{A\ (\gamma n,\gamma n)\text{-rigid}} (\Pi_A \otimes\Pi^A_{q(G)} \calS_1^{\otimes n^2}) \sum_{A' \in (\domain\cup \{\perp\})^{[A]}} \sum_{\substack{F \in \calF  \\  y \in \domain^F}}
    \beta_{A'}\beta^{A'}_{F,y}\ket{A'}_{[A]}\ket{y}_F 
    \ket{\perp}_{[B] \setminus F}}^2\nonumber\\
    &=\norm[\bigg]{\sum_{A\ (\gamma n,\gamma n)\text{-rigid}} \beta_A \ket{A}_{[A]} \otimes [\Pi^A_{q(G)} \calS_1^{\otimes n^2} \sum_{\substack{F \in \calF \\ |F|  \leq t \\  y \in \domain^F}}
    \beta^A_{F,y} \ket{y}_F 
    \ket{\perp}_{[B] \setminus F}]}^2\label{eq:matrix-mult-project-A}
\end{align}
for some $\beta_{A}$ and $\beta_{F,y}^A$ such that $\sum_{A \in (\domain \cup \{\perp\})^{n^2}} |\beta_{A}|^2 = 1$ and $\sum_{F \in \calF, y \in \domain^{F0}} |\beta_{F,y}^A|^2 = 1$ for each $A$.
Since $\Pi_{q(G)}^A$ only projects onto the $[B]$ input registers, each distinct choice of $\ket{A}_{[A]}$ gives orthogonal states so \cref{eq:matrix-mult-project-A} equals
\begin{align}
    &\sum_{A\ (\gamma n,\gamma n)\text{-rigid}} |\beta_A|^2 \norm[\bigg]{\Pi^A_{q(G)} \calS_1^{\otimes n^2} \sum_{\substack{F \in \calF \\ |F|  \leq t \\  y \in \domain^F}}
    \beta^A_{F,y} \ket{y}_F 
    \ket{\perp}_{[B] \setminus F}}^2 \nonumber\\
    &\leq \max_{A\ (\gamma n,\gamma n)\text{-rigid}} \norm[\bigg]{\Pi^A_{q(G)} \calS_1^{\otimes n^2} \sum_{\substack{F \in \calF \\  y \in \domain^F}}
    \beta^A_{F,y} \ket{y}_F 
    \ket{\perp}_{[B] \setminus F}}^2
    \label{eq:matrix-mult-rigid}
\end{align}  
\end{sloppypar}

We fix a $(\gamma n, \gamma n)$-rigid matrix $A$ that maximizes (\ref{eq:matrix-mult-rigid}) and partition the set $\calF$ based on the set
$\calL'(F)$ which contains all but at most
$\floor{\sqrt{k/2}}$ columns in $c(G)$.
Therefore we can rewrite (\ref{eq:matrix-mult-rigid})
as

\begin{align}
    \norm[\bigg]{\sum_{\substack{H \subseteq c(G)\\\text{s.t. }|H| \leq \floor{\sqrt{k/2}}}} \Pi^A_{q(G)}\calS_1^{\otimes n^2} \sum_{\substack{F \in \calF \\
    \calL'(F)=c(G)\setminus H\\  y \in \domain^F    \text{s.t. }\calL'(F)=c(G)\setminus H}}
     \beta^A_{F,y} \ket{y}_F 
    \ket{\perp}_{[B] \setminus F}}^2.
    \label{eq:matrix-mult-light}
\end{align}

Since $|c(G)|\le \min(k,n)$ we can upper bound (\ref{eq:matrix-mult-light}) by 

\begin{equation}
    \min(k,n)^{\sqrt{2k}}\cdot \max_{\substack{H \subseteq c(G)\\\text{s.t. }|H| \leq \floor{\sqrt{k/2}}}}\ \norm[\bigg]{ \Pi^A_{q(G)}\calS_1^{\otimes n^2} \sum_{\substack{F \in \calF \\   y \in \domain^F \\
    \text{s.t.}\calL'(F)=c(G)\setminus H}}
     \beta^A_{F,y} \ket{y}_F 
    \ket{\perp}_{[B] \setminus F}}^2.
    \label{eq:matrix-mult-H}
\end{equation}

We fix the set $H$ achieving the maximum value in \cref{eq:matrix-mult-H}, which fixes the value
of $\calL'(F)=c(G)\setminus H$.
This fixes the set 
$G^c_{\calL'(F)}$ of 
elements in $G$ that are in light columns of $F$ (equivalently,
not in $H$) which, since $F\in \calF$, contains at
least $k/4$ elements of $G$.
Let $G'$ be a fixed subset of $k/4$ of the elements of
$G^c_{\calL'(F)}$.  
By construction we have $c(G')\subseteq \calL'(F)$.
By only requiring that the outputs in $G'$ are correct,
we therefore can upper bound $\norm[\big]{\Pi_k \Pi_{\text{rigid}} \calS \ket{\phi_t}}^2$
by the maximum value of

\begin{equation}
    4 \min(k,n)^{\sqrt{2k}}\ 
    \norm[\bigg]{ \Pi^A_{q(G')}\calS_1^{\otimes n^2} \sum_{\substack{F \subseteq [B] \\  c(G')\ \subseteq\ \calL'(F)\\  y \in \domain^F}}
     \beta'_{F,y} \ket{y}_F 
    \ket{\perp}_{[B] \setminus F}}^2 \label{eq:matrix-mult-prelemma}
\end{equation}
over all $G'\subseteq [C]$ with $|G'|=k/4$ and $\beta'_{F,y}$ with $\sum_{F,y} |\beta'_{F,y}|^2=1$.

For each $j\in c(G')$, let $k_j$ be the number of elements
of $G'$ in column $j$.
Our overall strategy is to consider the 
$j\in c(G')$ one by one, and show that the total amplitude on states where these $k_j$
outputs are correct conditioned on the success for previous values of $j$ is of the form $d^{-\delta k_j}$ for
some fixed constant $\delta>0$.
These are $k_j$ outputs of the matrix-vector product
$Ay^{j}$ where $y^{j}$ is the $j$-th column of $B$ and
the fact that $c(G')\subseteq \calL'(F)$ implies that
$F$ has made at most $\beta\gamma n$ queries to $y^{j}$.
This is very similar to the situation with the matrix-vector problem from \cref{lem:matvec}.
In analogy with the \cref{lem:matvec}, we define $U^j$ to be the set of $k_j$ rows containing outputs of $G'$ in column $j$. 

Applying \cref{lem:partition} with $c=1$, for each $j\in c(G')$ there is a collection $V^{j}_1,\ldots, V^j_{\ell_j}$
of $\ell_j=\ceil{\gamma n/k_j}$ 
$k_j$-subsets of $[n]$ such that the $k_j\times k_j$ sub-matrix $A_{U^jV^j_i}$ has full rank.

Using the ideas of \cref{lem:matvec} we could bucket the possible basis states into one bucket for each large subset of the set associated with the
tuple $(V^j_{i_j})_{j\in q(G')}$ using \cref{lem:partition,lem:group} and bound each bucket separately.  
However, unlike its use in the proof of \cref{lem:matvec}, the value of many of the $k_j$ can
be very small, as low as 1, in which case the upper bounds
using \cref{lem:partition,lem:group} would yield a probability bound larger than 1.

Instead, we need a stronger argument that
depends on the amplitudes $\beta'_{F,y}$ in \cref{eq:matrix-mult-prelemma}.
The large subsets of the sets associated with tuples
 $(V^j_{i_j})_{j\in q(G')}$ yield
 candidate buckets but there are too many of them to be used.
 However, we will see in the following lemma that a relatively small collection of them
 can capture all but a constant fraction of the total amplitude given by the $\beta'_{F,y}$.   
 We will then see, in \cref{cor:matrix-prod-small-amp}, how this can be applied inductively with the portion of the total amplitude that
 is left over to yield a good upper bound on the total
 probability of producing the output values in $q(G')$, which 
 is what we need to prove.

\begin{lemma}\label{lem:matrix-prod-small-amp-recursive}
Let $G'\subseteq [C]$ with $|G'|=k/4$ and
$\calF'$ be a set of $F\subseteq [B]$ such that
$c(G')\subseteq \calL'(F)$.
Suppose further that $\sum_{F\in \calF',y\in \domain^F}
|\delta_{F,y}|^2=1$ for some $\delta_{F,y}$.
Define $\alpha=4\beta$. 
Then there is an $\calF''\subseteq \calF'$ and coefficients
$\delta'_{F,y}$ such that 
$\sum_{F\in \calF'', y\in \domain^F} |\delta'_{F,y}|^2=1$
and
\begin{equation}
\label{eq:bucket-reduction}
\norm[\big]{\Pi^A_{q(G')} \calS_1^{\otimes n^2} \sum_{\substack{F \in \calF' \\ y \in \domain^{F}}} \delta_{F,y} \ket{y}_F \ket{\perp}_{[B] \setminus F}}^2 \leq \frac{2^{1+H_2(\alpha)\,k / 2}}{d^{(1-\alpha)\,k / 4}} + \frac{1}{2}\norm[\big]{\Pi^A_{q(G')} \calS_1^{\otimes n^2} \sum_{\substack{F \in \calF'' \\ y \in \domain^{F}}} \delta'_{F,y} \ket{y}_F \ket{\perp}_{[B] \setminus F}}^2.
\end{equation}    
\end{lemma}

\begin{proof}
We first recall the definitions in our discussion
preceding the lemma statement.
For each $j\in c(G')$,
define $U^j$ to be the set of row indices of $G'$ in column $j$ and let $k_j=|U_j|$.
Define $\ell_j=\ceil{\gamma n/k_j}$, apply \cref{lem:partition} for each $j$, and let
$V^j_1,\ldots, V^j_{\ell_j}$ be the collection of
disjoint subsets of $[n]$ of size $k_j$ found for each $j$ such that each $k_j\times k_j$ sub-matrix $A_{U^j V^j_i}$ has full rank.

For each $F\in \calF'$ and $i\in c(G')$, define $F^j$ to be the set of 
row indices of elements of $F$ in column $j$;
since $c(G')\subseteq \calL'(F)$, we have
$|F^j|\le \beta\gamma n$.
For each $i\in [\ell_j]$ define
\begin{displaymath}
m^j_i=\sum_{F\in \calF',\ y\in \domain^F} |\delta_{F,y}|^2 \cdot |F^j\cap V^j_i|.
\end{displaymath}
Since $\sum_{F,y} |\delta_{F,y}|^2=1$, $m^j_i$ can be viewed
as the expected size of the overlap between the recorded
queries in the $j$-th column of the matrix $B$ and 
each $V^j_i$.
Since for each $j$, the sets $V^j_i$ are disjoint and
 $|F^j|\le \beta\gamma n$ we have
$\sum_{i\in [\ell_j]} m^j_i \le \beta\gamma n$.
Therefore, for each $j$, we have some index $i_j\in [\ell_j]$ such that
$m^j_{i_j}\le \beta\gamma n/\ell_j \le \beta k_j$.

Since $\sum_{j\in c(G')} k_j=|G'|=k/4$,
the expected total overlap between the recorded queries in
the columns of $G'$ and the chosen sets $V^j_{i_j}$ for
those columns is
$\sum_{j} m^j_{i_j}\le \sum_j \beta k_j =\beta k/4$.
Define $\calF''$ to be the set of $F\in \calF'$ such
that $\sum_{j} |F^j\cap V^j_{i_j}|\ge \alpha k/4 =\beta k$.
By Markov's inequality we have
\begin{equation}
\sum_{F\in \calF'',\ y\in \domain^F} |\delta_{F,y}|^2 \le
\frac{\sum_{j} m^j_{i_j}}{\beta k}\le 1/4.\label{eq:quarter-left}
\end{equation}
We split our analysis for $\calF'$ into two parts due to sets $F$ in $\calF''$ and
$\calF'\setminus\calF''$, respectively.

We begin with $F\in \calF''$.
Write
$\kappa=\sum_{F\in \calF'',\ y\in \domain^F} |\delta_{F,y}|^2 \le 1/4$.
For $F\in \calF''$, define $\delta'_{F,y}=\frac{1}{\sqrt{\kappa}}\delta_{F,y}$.
Then 
$\sum_{F\in \calF'', y\in \domain^F} |\delta'_{F,y}|^2=1$
and 
\begin{align}
\norm[\bigg]{\Pi^A_{q(G')} \calS_1^{\otimes n^2} \sum_{\substack{F \in \calF'' \\ y \in \domain^{F}}} \delta_{F,y} \ket{y}_F \ket{\perp}_{[B] \setminus F}}^2
&=\kappa \ \norm[\bigg]{\Pi^A_{q(G')} \calS_1^{\otimes n^2} \sum_{\substack{F \in \calF'' \\ y \in \domain^{F}}} \delta'_{F,y} \ket{y}_F \ket{\perp}_{[B] \setminus F}}^2\nonumber\\
&\le \frac{1}{4}  \norm[\bigg]{\Pi^A_{q(G')} \calS_1^{\otimes n^2} \sum_{\substack{F \in \calF'' \\ y \in \domain^{F}}} \delta'_{F,y} \ket{y}_F \ket{\perp}_{[B] \setminus F}}^2.
\label{eq:matrix-mult-F''}
\end{align}

We now consider $\calF'\setminus \calF''$.
By definition, for $F\in \calF'\setminus \calF''$,
we have $\sum_{j} |F^j\cap V^j_{i_j}|< \alpha k/4$.
By definition we have $\sum_{j} |V^j_{i_j}|=\sum_j k_j=k/4$
so $F$ must miss more than $(1-\alpha)k/4$ elements
of the set $V=\bigcup_j (V^j_{i_j}\times \set{j})$ of size $k/4$.
For each subset $V'$ of $V$ of size $k/4-\floor{\alpha k/4}$
we define a bucket $\calB_{V'}$ that contains sets $F$ that
must miss the elements of $V'$ and assign each $F\in \calF'\setminus\calF''$ to a unique bucket in an arbitrary fixed way.
There are at most $2^{H_2(\alpha)k/4}$ such buckets.
Then
\begin{align}
    &\norm[\bigg]{\Pi^A_{q(G')} \calS_1^{\otimes n^2} \sum_{\substack{F \in \calF' \setminus \calF'' \\ y \in \domain^{F}}} \delta_{F,y} \ket{y}_F \ket{\perp}_{[B] \setminus F}}^2 \nonumber\\
    &\leq \bigg(\sum_{\substack{V'\subseteq V\\ |V'|=k/4-\floor{\alpha k/4}}} \norm[\bigg]{\Pi^A_{q(G')} \calS_1^{\otimes n^2} \sum_{\substack{F \in \calB_{V'} \\ y \in \domain^{F}}} \delta_{F,y} \ket{y}_F \ket{\perp}_{[B] \setminus F}}\bigg)^2\nonumber\\
    &\leq 2^{H_2(\alpha)\,k / 2} \cdot\sum_{\substack{V'\subseteq V\\ |V'|=k/4-\floor{\alpha k/4}}} \norm[\bigg]{\Pi^A_{q(G')} \calS_1^{\otimes n^2} \sum_{\substack{F \in \calB_{V'} \\ y \in \domain^{F}}} \delta_{F,y} \ket{y}_F \ket{\perp}_{[B] \setminus F}}^2\nonumber\\
    &= 2^{H_2(\alpha)\,k / 2} \cdot\sum_{\substack{V'\subseteq V\\ |V'|=k/4-\floor{\alpha k/4}}} \norm[\bigg]{\Pi^A_{q(G')} \calS_1^{\otimes n^2}\ket{\perp}_{V'} \sum_{\substack{F \in \calB_{V'} \\ y \in \domain^{F}}} \delta_{F,y} \ket{y}_F \ket{\perp}_{[B] \setminus (F\cup V')}}^2
    \label{eq:matrix-mult-bucket}
\end{align}
where we first used the triangle inequality
followed by Jensen's inequality.

Now, applying the $\calS_1^{\otimes n^2}$ operator in (\ref{eq:matrix-mult-bucket}) will convert
the $\ket{\perp}_{V'}$ to a uniform superposition of 
all $\ket{y'}_{V'}$ for all $y'\in \domain^{V'}$ and convert
$\sum_{\substack{F \in \calB_{V'} \\ y \in \domain^{F}}} \delta_{F,y} \ket{y}_F \ket{\perp}_{[B] \setminus (F\cup V')}$
to some superposition of $\ket{y''}\in \domain^{[B]\setminus V'}$ with amplitudes some $\delta_{V',y''}$ such that
$\sum_{y''} |\delta_{V',y''}|^2 = \sum_{F\in \calB_{V'}, y\in \domain^F} |\delta_{F,y}|^2$.
Therefore, we can rewrite (\ref{eq:matrix-mult-bucket}) as
\begin{equation}
2^{H_2(\alpha)\,k / 2} \cdot \sum_{\substack{V'\subseteq V\\ |V'|=k/4-\floor{\alpha k/4}}} \norm[\bigg]{\Pi^A_{q(G')} \bigg[\sum_{y'\in \domain^{V'}} \frac{1}{\sqrt{d^{|V'|}}}\ket{y'}_{V'}\bigg] \otimes \sum_{y'' \in \domain^{[n]\setminus V'}} \delta_{V',y} \ket{y}_{[B] \setminus V'}}^2.
\label{eq:matrix-mult-after-S1}
\end{equation}

We now consider the application of $\Pi^A_{q(G')}$.
Let $V'_j\subseteq V^j_{i_j}$ be the set of row indices
in column $j$ of $V'\subseteq [B]$ and consider the corresponding set of columns in $A$.   
Since $A_{U^j V^j_{i_j}}$ has full rank, there is a subset
$U^j_0\subseteq U^j$ with $|U^j_0|=|V'_j|$ so that 
$A_{U^j_0 V'_j}$ also has full rank.
Now define $G'_0\subseteq G'$ to be $\bigcup_{j\in c(G)} (U_j\times \set{j})$ which has size $|V'|$.

For each $j$, the outputs in $U_j\times \set{j}\subset [C]$
can be expressed as the matrix-vector product $A_{U^j_0 V'_j} y^j_{V'_j}+ M'$ for some $|V'_j|\times |V'_j|$ matrix $M'$ defined by the product of the $U^j_0\times ([n]\setminus V'_j)$ submatrix of the fixed matrix $A$
and $y^j_{[n]\setminus V'_j}$. 
Since $A_{U^j_0 V'_j}$ is full rank, for each value of $M'$
given by $y^j_{[n]\setminus V'_j}$,
there is precisely one value of $y^j_{V'_j}$ that will
yield the output values $q(U_j\times \set{j})$.
Therefore, putting the properties for the columns of $c(G')$ 
together, there is precisely one value $y'\in \domain^{V'}$
that will yield the output values $q(G'_0)$.

It follows that, (\ref{eq:matrix-mult-after-S1}) is at most
\begin{align}
&2^{H_2(\alpha)\,k / 2} \cdot \sum_{\substack{V'\subseteq V\\ |V'|=k/4-\floor{\alpha k/4}}} \norm[\bigg]{\Pi^A_{q(G'_0)} \bigg[\sum_{y'\in \domain^{V'}} \frac{1}{\sqrt{d^{|V'|}}}\ket{y'}_{V'}\bigg] \otimes \sum_{y'' \in \domain^{[n]\setminus V'}} \delta_{V',y} \ket{y}_{[B] \setminus V'}}^2\nonumber\\
&=2^{H_2(\alpha)\,k / 2} \cdot \sum_{\substack{V'\subseteq V\\ |V'|=k/4-\floor{\alpha k/4}}} \norm[\bigg]{\frac{1}{\sqrt{d^{|V'|}}} \sum_{y'' \in \domain^{[n]\setminus V'}} \delta_{V',y} \ket{y}_{[B] \setminus V'}}^2\nonumber\\
&=2^{H_2(\alpha)\,k / 2} \cdot \sum_{\substack{V'\subseteq V\\ |V'|=k/4-\floor{\alpha k/4}}} \frac{1}{d^{|V'|}} \sum_{y'' \in \domain^{[n]\setminus V'}} |\delta_{V',y}|^2\nonumber \\
&=2^{H_2(\alpha)\,k / 2} \cdot \sum_{\substack{V'\subseteq V\\ |V'|=k/4-\floor{\alpha k/4}}} \frac{1}{d^{|V'|}} \sum_{F\in \calB_{V'}, y\in \domain^F} |\delta_{F,y}|^2\nonumber \\    
    &= 2^{H_2(\alpha)\,k / 2}\cdot \frac{1}{d^{|V'|}} \sum_{F \in \calF' \setminus \calF'', y \in \domain^F} |\delta_{F,y}|^2\nonumber\\
    &\leq 2^{H_2(\alpha)\,k / 2} / d^{(1-\alpha)\,k /4}\label{eq:matrix-mult-F'}
\end{align}
where the last equality follows since the buckets $\calB_{V'}$
partition $\calF'\setminus \calF''$.

We now combine the contributions from $\calF''$ and
$\calF'\setminus \calF''$.
Applying Jensen's inequality together with the bounds in (\ref{eq:matrix-mult-F''}) and (\ref{eq:matrix-mult-F'})  
we obtain that
\begin{align*}
&\norm[\big]{\Pi^A_{q(G')} \calS_1^{\otimes n^2} \sum_{\substack{F \in \calF' \\ y \in \domain^{F}}} \delta_{F,y} \ket{y}_F \ket{\perp}_{[B] \setminus F}}^2 \\
&\le
2\bigg[\norm[\big]{\Pi^A_{q(G')} \calS_1^{\otimes n^2} \sum_{\substack{F \in \calF'\setminus\calF'' \\ y \in \domain^{F}}} \delta_{F,y} \ket{y}_F \ket{\perp}_{[B] \setminus F}}^2 
+\norm[\big]{\Pi^A_{q(G')} \calS_1^{\otimes n^2} \sum_{\substack{F \in \calF'' \\ y \in \domain^{F}}} \delta_{F,y} \ket{y}_F \ket{\perp}_{[B] \setminus F}}^2\bigg] \\
&\le\frac{2^{1+H_2(\alpha)\,k / 2}}{d^{(1-\alpha)\,k / 4}} + \frac{1}{2}\norm[\big]{\Pi^A_{q(G')} \calS_1^{\otimes n^2} \sum_{\substack{F \in \calF'' \\ y \in \domain^{F}}} \delta'_{F,y} \ket{y}_F \ket{\perp}_{[B] \setminus F}}^2
\end{align*}
as required.
\end{proof}

\begin{corollary}\label{cor:matrix-prod-small-amp}
Let $G'\subseteq [C]$ with $|G'|=k/4$,
$\calF'$ be a set of $F\subseteq [B]$ such that
$c(G')\subseteq \calL'(F)$, and
$\sum_{F\in \calF',y\in \domain^F}
|\delta_{F,y}|^2=1$ for some $\delta_{F,y}$.
Then
    \begin{displaymath}
       \norm[\big]{\Pi^A_{q(G')} \calS_1^{\otimes n^2} \sum_{\substack{F \in \calF' \\ y \in \domain^{F}}} \delta_{F,y} \ket{y}_F \ket{\perp}_{[B] \setminus F}}^2 \leq 2^{2 + H_2(4 \beta)\, k / 2} / d^{(1-4\beta)\,k / 4}.
    \end{displaymath}
\end{corollary}

\begin{proof}
    Let $M$ be the maximum value of 
    \begin{displaymath}
        \norm[\big]{\Pi^A_{q(G')} \calS_1^{\otimes n^2}\mkern-20mu \sum_{F \in \calF',\ y \in \domain^{F}}\mkern-20mu \delta_{F,y} \ket{y}_F \ket{\perp}_{[B] \setminus F}}^2
    \end{displaymath}
    over all choices of $\calF'$ and $\delta_{F,y}$ with the
    required
    properties.
    This corollary follows from~\cref{lem:matrix-prod-small-amp-recursive}
    by observing that the right-hand term in \cref{eq:bucket-reduction} multiplied by
    $1/2$ is also upper bounded by $M$ and hence
   $M\le 2^{1 + H_2(4 \beta)\, k / 2} / d^{(1-4\beta)\,k / 4}+M/2$.
\end{proof}

Finally, plugging the bound from \cref{cor:matrix-prod-small-amp} into (\ref{eq:matrix-mult-prelemma}), we obtain that
the probability that $A$ and $B$ are both $(\gamma n,\gamma n)$-rigid and $\calC$ produces $k$ correct output values
for $C=AB$, $\norm[\big]{\Pi_k \Pi_{\text{rigid}} \calS \ket{\phi_t}}^2$, is at most
\begin{displaymath}
    16\ \min(k,n)^{\sqrt{2k}} \left(\frac{4^{H_2(4 \beta)}}{d^{(1-4\beta)}}\right)^{k/4} 
\end{displaymath}
as desired.
\end{proof}

\subsection{Related time-space tradeoff and cumulative memory lower bounds}
Now we use \cref{thm:mat-mult} to prove some related quantum linear algebra lower bounds.
Constructions of matching upper bounds can be found in \cref{sec:query-algs}.

\begin{corollary}\label{cor:mat-square}
    Let $\F$ be a field and $\domain \subseteq \F$ with $d= |\domain|$. If $\calC$ is a quantum circuit that computes the function $f: \domain^{n^2} \to \F^{n^2}$ where $f(A) = A^2$ on all upper triangular inputs in time $T$ and space $S$ 
    with success probability at least $1/T$, then $T$ must be $\Omega(n^3 \sqrt{\log d\ /S})$.
\end{corollary}

\begin{proof}
    Let $A,B \in \domain^{n^2}$ and construct the $3n \times 3n$ matrix:
    \begin{equation*}
        M = \begin{bmatrix} 0 & A & 0 \\ 0 & 0 & B \\ 0 & 0 & 0\end{bmatrix}
    \end{equation*}
    Since the top right $n \times n$ sub-matrix of $M^2$ is equal to the product $AB$, we get a reduction from matrix multiplication and can apply \cref{thm:mat-mult} to derive the lower bound.
\end{proof}

Using \cref{prop:TS-to-cm-lb} we can also bound the cumulative memory complexity for these problems.
\begin{corollary}
    Let $\F$ be a field and $\domain \subseteq \F$ with $d=|\domain|$. If $\calC$ is a quantum circuit that computes the function $f: \domain^{2n^{2}} \to \F^{n^2}$ given by $f(A,B) = AB$ or the function $g: \domain^{n^2} \to \F^{n^2}$ given by $f(A) = A^2$, then $\calC$ must have cumulative memory complexity $\Omega(n^6 \log (d) \ / T)$.
\end{corollary}
\begin{proof}
    For $f$, we apply \cref{prop:TS-to-cm-lb} with \cref{lem:mat-mul-tail} where $m'$ is $\Theta(n^2)$, $\Delta$ is $1/2$, $h_1(n)$ is $\Theta(n)$, $K(n) = d^{-1/48}$, $C=16$.
    This gives us that the cumulative memory complexity is $\Omega(n^6 \log (d) \ / T)$.
    Using the same reduction as in \cref{cor:mat-square}, this same lower bound applies to computing $g$.
\end{proof}
\section{Quantum tradeoffs for Boolean
matrix operations}\label{sec:quantum-boolean}

In this section we focus on Boolean matrix operations, which
use $(AND,OR)$ inner product of vectors rather than 
the usual $(+,\times)$ inner
product.
We denote this Boolean inner product of vectors $u$ and $v$  by $u\bullet v$ and extend this notation to Boolean matrix-vector
product and Boolean matrix multiplication.
For $u,v\in \set{0,1}^n$, $u\bullet v=1$ if and only if the subsets of $[n]$ encoded by $u$ and $v$ intersect, so
the problems of computing Boolean matrix multiplication
and Boolean matrix-vector product can be seen as computing many
correlated copies of the set disjointness problem. 

\subsection{Tradeoffs for Boolean matrix multiplication}

Unlike what we have shown for algebraic problems, as noted in~\cite{KSdW07}, quantum algorithms for Boolean matrix multiplication have
better time-space tradeoff properties than their classical counterparts.

\begin{proposition}
\label{prop:grover-matrix}
  For any $c>0$, there are quantum circuits computing $n\times n$ Boolean matrix multiplication $A\bullet B$ with error at most $n^{-c}$ using space $O(\log n)$ and a number of queries $T$ that is $O(n^{2.5}\log n)$.
\end{proposition}

\begin{proof}
    Fix $c>0$.  
    Each
    of the $n^2$ entries in the product is a disjointness function of length $n$ that can be computed with error at most
    $n^{-c-2}$ and space $O(\log n)$ using Grover's algorithm in time $O(\sqrt{n}\log n)$
    for error at most $n^{-c}$ overall.
\end{proof}

This is in contrast to the following result of Abrahamson which shows that classical algorithms as fast as this quantum algorithm
require space $\tilde \Omega(n^{0.5})$ rather than
$O(\log n)$.

\begin{proposition}[\cite{DBLP:conf/focs/Abrahamson90}]
\label{prop:abrahamson=boolean}
There is a probability distribution on input matrices and
constants $0<c_1<c_2$ under which the best classical algorithms (branching programs) for Boolean matrix multiplication $A\bullet B$ using time $T$ and space $S$ require
$T\cdot S$ that is
$\begin{cases}
    \Theta(n^{3.5})&\textrm{for }T\le c_1 n^{2.5}\\
    \Theta(n^3)&\textrm{for }T\ge c_2 n^{2.5}.
\end{cases}$
\end{proposition}


For quantum circuits, Klauck, \v{S}palek, and de Wolf~\cite{KSdW07} proved the following 
time-space tradeoff lower bound which proves that the
quantum algorithm in \cref{prop:grover-matrix} is nearly optimal when the space
$S$ is $O(\log n)$.

\begin{proposition}[\cite{KSdW07}]
\label{prop:ksdw}
Any bounded error quantum circuit that computes the $n \times n$ Boolean matrix multiplication $A\bullet B$ with $T$ queries and space $S$ requires $T$ to be $\Omega(n^{2.5}/S^{0.5})$.
\end{proposition}

A key difference between the methods used in Abrahamson's bounds and our results for linear algebra versus those in this proof
is that we require that the set of output values produced in each
part of the computation is fixed independent of the input. 
(See our discussion of such output-oblivious computation in~\cref{subsec:prelim-ts-tradeoffs}.)
Such an assumption was essential for the quantum time-space lower bounds in \cite{KSdW07, ASdW09}, although the bound for multiple disjoint collision pairs in \cite{HM21} and our results in \cref{sec:mat-vec,sec:mat-mult} apply to quantum query algorithms without such a restriction on output production.
Fixing the output values produced in each part of the computation allows one to go beyond using a single hard distribution
on inputs, and instead choose hard distributions for each part of the computation
depending on the target outputs.
To give a sense of how this works we sketch the lower bound method of~\cite{KSdW07} for Boolean matrix multiplication, which relies on a strong direct product lemma for the function $OR_n^k$ (i.e. $k$ independent copies of the $OR$ function each on inputs of size $n$):

\begin{proposition}[Strong Direct Product Theorem for $OR^k_n$~\cite{KSdW07}]
\label{prop:SDPT}
There are positive constants $\varepsilon$ and $\gamma$ such that the following
hold:
\begin{description}
    \item{(a)}   Any randomized algorithm making at most $\varepsilon kn$ queries
    has success probability at most $2^{-\gamma k}$ in computing $OR^k_n$.
    \item{(b)}   Any quantum algorithm making at most $\varepsilon k\sqrt{n}$ queries
    has success probability at most $2^{-\gamma k}$ in computing $OR^k_n$.
\end{description}
\end{proposition}

\begin{proof}[Proof sketch for \cref{prop:ksdw}]
For any integer $k\le n/2$, the function $OR^{k}_{\floor{n/k}}$ can be embedded in any set $E\subseteq [n]\times [n]$ of $k$ outputs of the $n\times n$ Boolean matrix product
$A\bullet B$ as follows:
 Begin by dividing $[n]$ into $k$ blocks $b_1,\ldots,b_k$ each of size $\floor{n/k}$ (together with at most $k-1$ other elements) and associate each $(i,j)\in E$,
 with a distinct index $\ell=\ell(i,j)\in [k]$.
For each $(i,j)\in E$, for $\ell=\ell(i,j)$ set every entry in $A_{i,b_\ell}$ to 1 and set the vector of inputs in $B_{b_\ell,j}$ to the
$\ell$-th block of the input to $OR^{k}_{\floor{n/k}}$.
Set all other bits in $A$ and $B$ to 0.
It is easy to see that the $k$ outputs indexed by $E$ will be
the outputs for $k$ disjoint $OR$ functions on $\floor{n/k}$ bits.

Without loss of generality one can assume that the space bound $S$ is at most $\alpha n$ for some small constant $\alpha>0$ since the number of queries must be $\Omega(n^2)$ in the worst case\footnote{Note that this is not completely obvious since quantum algorithms for some problems may have a sublinear number of queries.}.
Choose $k=cS$ for some suitably large constant $c$ that 
depends on the constant $\gamma$ in \cref{prop:SDPT}. 
Begin by slicing the circuit into layers of $\varepsilon\sqrt{kn}$
queries each.   
There are $\Theta(T/\sqrt{kn})$ such layers.
By \cref{prop:SDPT} and the embedding, any circuit of depth $\varepsilon \sqrt{kn}=\varepsilon k\sqrt{n/k}$ queries can produce $k$ correct output values with
probability only $2^{-\gamma k}$ for some $\gamma>0$.
This is the same depth as each of the layers but each layer also gets an $S$ qubit input-dependent state to begin.
By~\cref{prop:quant-union}, the probability that the resulting
layer can produce $k$ correct output values is at most $2^{S}2^{-\gamma k}$ which is at most $2^{-S}$ if the constant $c$ used in defining $k$ is sufficiently large.

Therefore, the total number of correct output values that can be produced
with probability larger than $2^{-S}$ must be 
$O(T/\sqrt{kn})\cdot k$ which is $O(T\sqrt{S/n})$.
On the other hand this number of outputs produced must be at least $n^2$.
It follows that $T$ must be $\Omega(n^{2.5}/\sqrt{S})$.
\end{proof}

\subsubsection*{Our improved lower bound}

\begin{theorem}
\label{thm:quantum-booolean-matrix}
Any quantum circuit computing $n\times n$ Boolean matrix multiplication $A\bullet B$ with $T$ queries and space $S$ and success probability more than $2^{-S}$
must have 
 $T$ that is $\Omega(n^{2.5}/S^{1/4})$.
\end{theorem}

Though the form of our lower bound may seem somewhat unusual, both the exponent of $n$ and that of $S$ are optimal:  
The algorithm of \cref{prop:grover-matrix} shows that exponent of $n$ is optimal since there is only a gap of $O(\log^{5/4} n)$ for space $\Theta(\log n)$. 
In our quantum query model, at the other end of the scale, an algorithm with space $3n^2$ 
can query and completely remember both matrices in $2n^2$ time and $2n^2$ space, after which a single global unitary transformation will produce the
$n^2$ bits of output needed in the remaining $n^2$ qubits of working memory; hence the exponent of $1/4$ on $S$ cannot be reduced. 

\cref{thm:quantum-booolean-matrix} follows from the following key lemma which improves on the corresponding bound in~\cite{KSdW07} by a factor of $\Theta(k^{1/4})$.

\begin{lemma}
\label{lem:main-quantum-boolean}
There are constants $\varepsilon,\gamma>0$ such that the following holds.
Let  $k< n^2/100$  be an integer.
For any quantum circuit $\calC$ with at most $\varepsilon k^{3/4} n^{1/2}$ 
queries to $x$, the probability that $\calC$ produces $k$ correct output values of $n\times n$ Boolean matrix multiplication $A\bullet B$ is at most $2^{-\gamma k}$.
\end{lemma}

We first see how this lemma suffices for the theorem:

\begin{proof}[Proof of \cref{thm:quantum-booolean-matrix} via \cref{lem:main-quantum-boolean}]
Since there are $n^2$ outputs, it seems that $T\ge n^2$ queries 
are required, but that isn't quite obvious. 
Nonetheless, we can, for example, derive a $T=\Omega(n^2)$ lower bound by applying~\cref{lem:main-quantum-boolean} with $k=n^2/101$ which shows that a circuit with at most some $\beta n^2$ 
queries can only achieve exponentially small success probability for producing a small fraction of the output.
Therefore without loss of generality we can assume that $\sqrt{S}< \alpha n$ for some arbitrarily small constant $\alpha>0$. 
Let $\varepsilon$ and $\gamma$ be the constants from \cref{lem:main-quantum-boolean}.
Let $c=2/\gamma$ and define $k=cS$. 
Therefore for $\alpha\le  1/(10\sqrt{c})$ we obtain
that $5\sqrt{k}=5\sqrt{cS}< n/2$.
By \cref{lem:main-quantum-boolean}, since $k< n^2/100$, any quantum
query algorithm with at most $\varepsilon k^{3/4}n^{1/2}$ queries
has success probability at most $2^{-\gamma k}=2^{-2S}$ of producing $k$ correct output values. 

We prove the contrapositive of the theorem statement:  Suppose that $T\le \varepsilon n^{2.5}/(cS)^{1/4}=\varepsilon n^{2.5}/k^{1/4}$.
When we divide
$\calC$ into layers with $\varepsilon k^{3/4}n^{1/2}$
quantum queries each, there are at most $n^2/k$ layers.
Since there are a total of $n^2$ outputs, there must be some layer $i$
during which at least $k$ outputs are produced.
Let $E$ be the set of the first $k$ outputs produced in layer $i$.
By the argument above
since the space is at most $S$, by \cref{prop:quant-union}
the probability that 
 these $k$ outputs are correct given the $S$ qubits of input-dependent initial state at the beginning of layer $i$ is at most $2^{S}$ times larger than that of a circuit without them and the
 same number of queries, which is at most 
 $2^{S}\cdot 2^{-2S}=2^{-S}$ which is what we needed to show.   
\end{proof}

The main idea behind the proof of this key lemma is
an improved method for embedding the direct product of $OR$
functions into outputs of the Boolean matrix multiplication
problem;
this uses the following definition of an $L$-coloring
of subsets of $[n]\times [n]$.

\begin{definition} 
\label{dfn:L-coloring}
For $E\subseteq [n]\times [n]$ an $L$-\emph{coloring} of
$E$ is a map $\chi:E\rightarrow [L]$ such that
\begin{itemize}
    \item 
within each color class either all rows are distinct or all columns are distinct, and
\item for each color $\ell$ there is a rectangle given by sets $R_\ell\subseteq [n]$ of
rows and $C_\ell\subseteq[n]$ of columns such that the set of
points of color $\ell$ is precisely $E\cap (R_\ell\times C_\ell)$.  
\end{itemize}    
(Note that the rectangles $R_\ell\times C_\ell$ may overlap, but their overlap must not
contain any points in $E$, see \cref{fig:coloring-example}.)
\begin{figure}[t]
    \centering
    \includegraphics[width=0.7\linewidth]{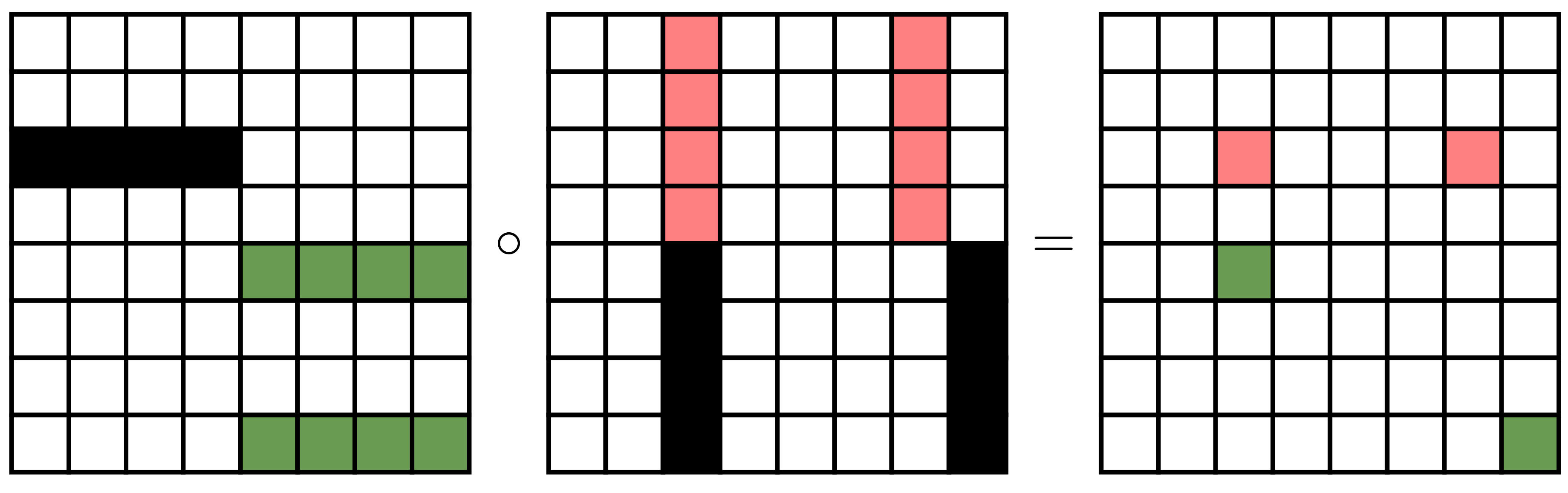}
    \caption{An example of a valid 3-coloring (as in \cref{dfn:L-coloring}), where the pink and green squares on the right matrix correspond to the colored outputs. For the left two matrices, the black squares are fixed to the input $1$ while the white square are fixed to the input $0$. The pink and green squares in the left two matrices encode an input to $OR_4^4$ whose outputs are the colored entries of the right matrix.}
    \label{fig:coloring-example}
\end{figure}

We say that a rectangle $R\times C\in [n]\times [n]$ is \emph{colorable}
iff $E\cap (R\times C)$ either has all its elements in different rows or
all its elements in different columns.
\end{definition}

The motivation for this definition is given by the following
lemma.

\begin{lemma}
\label{lem:coloring-or}
    Let $E\subseteq [n]\times [n]$ with $|E|=k$ and $L$ be an integer with $L\le n/2$.
    If $E$ has an $L$-coloring then $OR^k_{\floor{n/L}}$ is a sub-function of the
    function that produces the $k$ outputs of 
    $A\bullet B$ indexed by $E$ for $n\times n$ Boolean matrices $A$ and $B$.
\end{lemma}

\begin{proof}
Write $E=\dot{\bigcup}_{\ell=1}^L E_\ell$ where $E_\ell$ is the set of $(i,j)$ in $E$ in color class $\ell$.
We now divide $[n]$ into $L$ disjoint blocks $b_1,\ldots, b_L$
of at least $\floor{n/L}\ge 2$ elements each.
Given the coloring and division into blocks, we define a partial assignment to the matrices $A$ and $B$ as follows:

\begin{itemize}
\item If color class $\ell$ consists of points that do not share a
column, for each $(i,j)\in E_\ell$, we set all entries of 
$A_{i,b_\ell}$ to 1 and leave all entries of $B_{b_\ell,j}$ unset.

\item If color class $\ell$ consists of points that do not share a
row, for each $(i,j)\in E_\ell$, we set all entries of
$B_{b_\ell,j}$ to 1 and leave all the entries
of $A_{i,b_\ell}$ unset.

\item All entries of $A$ and $B$ that are not defined by the above 
two cases are set to 0.
\end{itemize}

In particular, this means that if $E_\ell$ does not contain
any element of the form $(i,\cdot)$ then the submatrix $A_{i,b_\ell}$ is all 0 and
if $E_\ell$ does not contain any element of the form $(\cdot,j)$
then the submatrix $B_{b_\ell,j}$ is all 0.

It remains to show
that the outputs in $E$ of this matrix product are $k$ disjoint ORs on at least $\floor{n/L}$ bits each.

Observe that if the color of $(i,j)$ is $\ell$, there cannot be another color $\ell'\neq \ell$ and $i' \neq i,\ j'\neq j$ such that $(i, j'), (i',j) \in E$ both have color $\ell'$, as this would violate the rectangle condition for color $\ell'$. 
This implies that either all entries of $A_{i, b_{\ell'}}$ are $0$ or all entries of $B_{b_{\ell'}, j}$ are $0$ for all $\ell' \neq \ell$. 
Therefore, assuming that $(i,j)$ is colored $\ell$, the $(i,j)$ entry of the product must equal $A_{i, b_\ell} \bullet B_{b_\ell, j}$.

If color class $E_\ell$ consists of points that do not share a column then
 the output for each $(i,j)\in E_\ell$  is the $OR$ of the 
$\ge \floor{n/L}$ unrestricted input bits of $B_{b_\ell,j}$;
the inputs for different $(i,j)$ are disjoint since no two points of $E_\ell$ share a column.
The analogous property holds for each color class 
$E_\ell$ whose points do not share rows.  
In that case, each output $(i,j)\in E_\ell$
is the $OR$ of $\ge \floor{n/L}$ unrestricted input bits of $A_{i,b_\ell}$ and input bits of $A_{i,b_\ell}$ are disjoint from each other.
Finally, the disjointness of the inputs to the $OR$ functions associated with
different color classes is inherited from the disjointness of 
$b_1,\ldots, b_L$, and the lemma  follows since $|E|=k$.
\end{proof}

The lower bound of~\cite{KSdW07} in \cref{prop:ksdw} embedded
$OR^k_{\floor{n/k}}$ into any set $E$ of $k$ outputs of $A\bullet B$.
Their argument corresponds to the trivial $k$-coloring
that assigns each element of $E$ to its own color class.

\begin{definition}
    For integer $k>0$ define $L_\alpha(k)$ to be the minimum number
    of colors $L$ such that for all subsets $E\subseteq [n]\times [n]$ with $|E|\le k$, there is an $L$-coloring of a subset $E'\subseteq E$ with $|E'|\ge \alpha |E|$.
\end{definition}

\begin{lemma}
\label{lem:general-quantum-boolean}
There are constants $c,c'>0$ such that the following holds.
Let $\alpha>0$ and $k$ be an integer such that $L_\alpha(k)\le n/2$.
For any quantum circuit $\calC$ with at most $c k n^{1/2}/L_{\alpha}(k)^{1/2}$ 
queries to $x$, the probability that $\calC$ produces $k$ correct output values of $n\times n$ Boolean matrix product $A\bullet B$ is at most $2^{-c'\alpha k}$.
\end{lemma}

\begin{proof}
Let $E$ be any fixed set of $k$ output positions in
$A\bullet B$.
We show that for each fixed value of $E$ the probability that
$\calC$ can correctly guess the output values at these indices is exponentially small in $k$.
Let $L\le L_\alpha(k)$ be such that there is an $L$-coloring of a subset $E'\subseteq E$ with $|E'|\ge \alpha |E|$.
By \cref{lem:coloring-or}, $OR^{\lceil\alpha k\rceil}_{\floor{n/L}}$ is a sub-function
of the $\lceil\alpha k\rceil$ outputs indexed by the set $E'$.
Since $L\le n/2$, $\floor{n/L} \ge 2n/(3L)$ and $\sqrt{\floor{n/L}}\ge 4\sqrt{n/L}/5.$
Choose $c=4\varepsilon\alpha/5$ and $c'=\gamma$ for $\varepsilon$ and $\gamma$ given in~\cref{prop:SDPT}.
By that proposition, the probability that $\calC$ produces the values of these $k$
outputs correctly is at most the probability that $\calC$ produces the $\lceil\alpha k\rceil$ outputs in $E'$ correctly which is $2^{-\gamma \lceil\alpha k\rceil}\le 2^{-c' \alpha k}$.
\end{proof}

Then~\cref{lem:main-quantum-boolean} is an immediate corollary of
\cref{lem:general-quantum-boolean} and
the following bound on $L_{1/2}(k)$.

\begin{lemma}[Coloring Lemma\protect\footnote{In a preliminary version of this paper~\cite{DBLP:conf/stoc/BeameKW24} there was an error in this lemma, which claimed to show that $L_1(k) \leq 2 \sqrt{6k}$.
    We thank the anonymous reviewers for asking the question that led to us find and address this error.}]\label{lem:coloring}    
    $L_{1/2}(k)\le 2\sqrt{6k} < 5 \sqrt{k}$.
\end{lemma}

\begin{proof}
Without loss of generality, $E$ is contained in a grid with side lengths at least $ n > 2\sqrt{6k}$, as otherwise we could just use a single color for each row (or column).
For a given subset $A \subseteq [n]$ or rows or columns, we use $\overline{A}$ to denote $[n] \setminus A$.

Our strategy is as follows: for some constant $c$ to be determined we show that either
\begin{enumerate}
    \item there is a row containing at least $c\sqrt{k}$ points of $E$, or
    \item there is a rectangle $R\times C$ such that there are at least $c \sqrt{k}$ points in the rectangle, all of which can be colored with a single color.
    Moreover, in this case, we show that $|(\overline{R}\times C)\cap E|\leq|(R\times C)\cap E|$.
\end{enumerate}
We now argue why the above two conditions are enough to prove that $L_{1/2}(k) \leq \frac{2}{c}\sqrt{k}$.

If we colored a single row or column, then we can inductively color the remaining points of $E'\subseteq E$ outside that row/column with no issue. 
However, if we colored the points in $R\times C$, inductively coloring the remaining points could cause an issue because of the rectangle requirement for colors.
 To address this, we discard the points of $(\overline{R}\times C)\cap E$ and proceed inductively on $E':= E \cap ([n] \times \overline C)$.
At the end of the procedure, since we always color at least the number of points we discard, we will have discarded at most $k/2$ points, as desired.

It remains to show that this such a coloring would always use at most $\frac{2}{c}\sqrt{k}$ colors. We prove this using induction. 
Indeed, applying induction to color at least $1/2$ of the
remaining $k'\leq k-c\sqrt k$ elements of $E'$ in $[n]\times \overline C$
will require at most
$\frac{2}{c}\, \sqrt{k'}=\frac{2}{c}\, \sqrt{k-s} \leq \frac{2}{c} \sqrt{k}(1-\frac{2s}{k}) \leq \frac{2}{c}\sqrt{k} - 1$ colors.
It follows that at
most $\frac{2}{c}\, \sqrt{k}$ colors are needed to color at least 1/2 the points in $E$, as required.

We now show that we can execute this strategy with the constant $c=1/\sqrt{6}$, which will prove the lemma. That is, we show how to find either a row containing at least $\sqrt{k/6}$ points of $E$ or a colorable rectangle $R\times C$ with at least $\sqrt{k/6}$ points of $E$ such that $|E\cap(\overline R \times C)| \leq |E\cap (R \times C)|$.

For any column $j$ we write $E^j$ for the set of $i$ such that $(i,j)\in E$.   
Build $R\times C$ in the following way:\footnote{In \cref{alg:finding-colorable-rectangle}, instead of the constant 3/4 in \cref{line:while}, we could have chosen any $(1-\gamma)$ instead. In this case, we would achieve a bound for $L_{1-2\gamma}(k) \leq 2\sqrt{\frac{1-\gamma}{\gamma(1-2\gamma)} k}$. For simplicity, we have chosen $\gamma=1/4$, which is quite close to optimal and has a larger value of $\alpha = 1-2\gamma$.}

\begin{algorithm}[ht]
    \caption{Finding a colorable rectangle with many points.}
    \label{alg:finding-colorable-rectangle}
      Initialize $R\leftarrow \varnothing$;\, $C\leftarrow \varnothing$;\, $D\leftarrow \varnothing$\\
       \While{there is a $j$ such that $|E^j\setminus (R\cup D)|\ge \frac{3}{4}|E^j|$ \label{line:while}}{
            $C\leftarrow C\cup \{j\}$\label{line:C-update}\\
            $D\leftarrow D\cup (R\cap E^j)$ \label{line:D-update}\\
            $R\leftarrow (R\setminus E^j)\cup (E^j\setminus  D)$ \label{line:R-update}
        }
        
    \Return{$R\times C$}
    
    \end{algorithm}
\begin{figure}[t]
    \centering
    \includegraphics[width=0.3\linewidth]{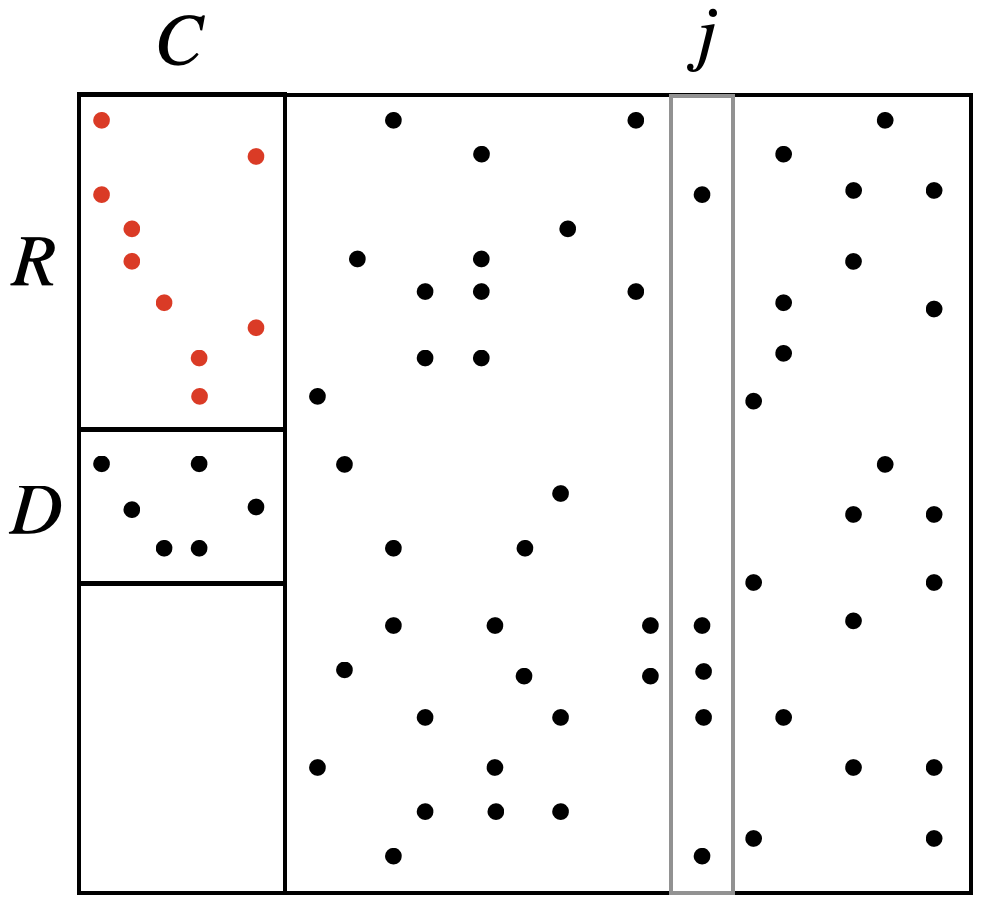}
    \caption{Visualization of a single iteration of \cref{alg:finding-colorable-rectangle}.}
    \label{fig:coloring-alg-visualization}
\end{figure}

First, observe that at the end of the procedure (and indeed at the end of every iteration) the rectangle $R\times C$ contains
exactly one element of $E$ in every row, every row of $D\times C$
contains at least two elements of $E$, and there are no
elements of $E$ in $\overline{(R\cup D)}\times C$ -- see \cref{fig:coloring-alg-visualization} for a visualization of these observations.

Our first simple claim lets us bound the number of points in $\overline{R}\times C$.

\begin{claim}
    $|E \cap (\overline R \times C)| \leq |E\cap (R\times C)|$,
    and
    $|D| \leq |R|/2$.
\end{claim}

\begin{proof}[Proof of Claim]
    The claim is true initially. 
    Suppose that it is true at the beginning of an iteration.
    When we add $j$ to $C$ on \cref{line:C-update},  we have
    $|E^j\setminus (R\cup D)|\ge 3|E^j|/4$, and therefore have
    $|R\cap E^j|\le|E^j|/4$.

    \cref{line:D-update} therefore adds at most $|E^j|/4$ row indices to
    $D$.
    Since each element of $R\times C$ contained exactly one element of $E$
    at the end of the previous iteration, each row added to $D$
    by \cref{line:D-update} has exactly two points of $E$ in the columns of $C$ and there are no points of $E$
    in $\overline{(R\cup D)}\times C$,  the iteration adds at most $2|E^j|/4=|E^j|/2$ points of $E$ to $\overline{R}\times C$.
    
    On the other hand, \cref{line:R-update} adds at least $3|E^j|/4$ elements of $E^j$ to $R$ and only removes the at most $|E^j|/4$ elements
    of $R\cap E^j$, so $R$ grows
    by at least $|E^j|/2$ rows in total.
    Since each row
    of $R\times C$ has exactly one point
    in the columns of $C$, at
     least $|E^j|/2$ points of $E$ get added to $R \times C$.

    Counting rows, we have added
    at most $|E^j|/4$ rows to $D$ and at
    least $|E^j|/2$ rows to $R$, which maintains that $|D|\le |R|/2$.

    Counting points, the increase in size of
    $E\cap (\overline R \times C)$ is at most $|E^j|/2$
    which lower bounds the
     net gain for $E\cap (R \times C)$.
     This maintains
     $|
    E\cap (\overline R \times C)|\le |E\cap (R \times C)|$ as required.
\end{proof}

We let $s$ be the larger of $|R|$ and the maximal number of points in $E$ of any row. 
For convenience, write $Z=R\cup D$.

When \cref{alg:finding-colorable-rectangle} finishes, for every column $j\in \overline C$, fewer than $3/4$ of its points are in rows of $\overline Z$ and
hence more than $1/4$ of its points are in rows of $Z$.
So we must have that
\begin{displaymath}
|E\cap (Z\times \overline C)|>
|E\cap (\overline Z\times \overline C)|/3.
\end{displaymath}
As $\overline Z\times C$ has no points of $E$ and each
row has at most $s$ points of $E$,
the total number of points is
\begin{align*}
    k &=|E\cap (Z\times [n])|+|E\cap (\overline Z\times [n])|\\
    &=|E\cap (Z\times [n])|+|E\cap (\overline Z\times \overline{C})|\\
    &\le |E\cap (Z\times [n])| + 3|E\cap (Z\times \overline{C})|\\
    &\leq 4|Z|\,s \leq 4\cdot (3|R|/2)\, s \leq 6s^2.
\end{align*}
Therefore $s\geq\sqrt{k/6}$.
\end{proof}

\cref{lem:main-quantum-boolean} is an immediate corollary of \cref{lem:general-quantum-boolean,lem:coloring} which
completes the proof of \cref{thm:quantum-booolean-matrix}.

\cref{thm:quantum-booolean-matrix} can be directly extended to an equivalent lower bound on the quantum cumulative memory complexity for Boolean matrix multiplication.

\begin{corollary}
    Any quantum circuit computing $n \times n$ Boolean matrix multiplication $A \bullet B$ with $T$ queries, space $S$, and success probability more than $1/(2T)$ must have cumulative memory that is $\Omega(n^{10}/T^3)$
\end{corollary}

\begin{proof}
    Using \cref{lem:general-quantum-boolean,lem:coloring}, we can apply \cref{prop:TS-to-cm-lb} with $C=1$, $m'(n) = n^2/8$, $h(k,n)=ck^{3/4}n^{1/2}/2^{1/4}$, $K(n) = 2^{c'}$ where constants $c,c'>0$. This gives us a cumulative memory lower bound of:
    \begin{equation*}
        \Omega(\min(n^{10}/T^3, n^{4})) = \Omega(n^{10}/T^3)
    \end{equation*}
    as $T$ must be $\Omega(n^2)$.
\end{proof}

We also obtain a general classical lower bound from these arguments. We start by showing a classical analogue of \cref{lem:general-quantum-boolean}.

\begin{lemma}
\label{lem:general-classical-boolean}
Let $\varepsilon,\gamma>0$ be the constants from \cref{prop:SDPT}.
Let  $k$ be an integer such that $L(k)\le n/2$.
Any randomized algorithm with at most $(2\varepsilon/3) k n / L(k)$ 
queries to $x$ can only produce $k$ correct output values of $n\times n$ Boolean matrix product $A\bullet B$ with probability at most $2^{-\gamma k}$.
\end{lemma}
\begin{proof}
    Let $E$ be any fixed set of $k$ output indices in $A \bullet B$.
    Let $L \leq L(k)$ be the smallest number such that $E$ can be colored with $L$ colors.
    By \cref{lem:coloring-or} we know that $OR^k_{\floor{n/L}}$ is a sub-function of the outputs indexed by $E$.
    Thus, by \cref{prop:SDPT} any randomized algorithm making at most $\varepsilon k \floor{n/L} \geq (2 \varepsilon / 3) k n / L(k)$ queries can compute these outputs with probability at most $2^{-\gamma k}$.
\end{proof}

\begin{theorem}
\label{thm:classical-boolean-matrix}
Any output-oblivious classical query algorithm computing $n\times n$ Boolean matrix-multiplication with $T$ queries and space $S$ with success probability more than $2^{-S}$
must have 
 $T$ that is $\Omega(n^{3}/\sqrt{S})$.
\end{theorem}

\begin{proof}
    Since there are $n^2$ outputs, which is a trivial time lower bound for sequential algorithms, we can assume that 
    $\sqrt{S}$ is at most $\alpha n$ for some arbitrarily small constant
    $\alpha>0$.
    Let $c=2/\gamma$ for $\gamma$ given by \cref{prop:SDPT} and let $k=cS$.   
    Our assumption with $\alpha<1/(10\sqrt{c})$ implies, by~\cref{lem:coloring} that $L(k)<5\sqrt{k}=5\sqrt{cS}< n/2$.
    The main difference in parameters from the quantum case is
    that we need to apply \cref{lem:general-classical-boolean} instead of \cref{lem:general-quantum-boolean}  to say that classical output-oblivious branching programs of width $2^S$ have success probability at most $2^{-\gamma k}=2^{-2S}$ of computing $k$ correct output values of $A \bullet B$.
    There are at most $2^S$ nodes
    at a layer boundary and hence the probability that a layer of
    height $(2\varepsilon/3) kn/L(k)$ correctly
    produces $k$ output values is at most $2^{-S}$.
    Rewriting using $L(k)< 5\sqrt{k}$,
    we obtain that a layer of height $(2\varepsilon/15) \sqrt{k}\,n$ correctly produces outputs with probability at most $2^{-S}$.
    Since there are $n^2$ outputs, for any circuit of depth $T$ at most $(2 \varepsilon/15) n^3/\sqrt{k}$ must have some layer of depth $2 \varepsilon/15) \sqrt{k}\, n$ during which at most $k$ outputs are produced and each output value must be correct for the algorithm to be correct, so the overall success probability is at most $2^{-S}$.
\end{proof}

This achieves the goal suggested by 
Klauck, \v{S}palek, and de Wolf~\cite{KSdW07} who ventured that the likely tight tradeoff for classical computation of Boolean matrix multiplication is $T^2 S=\Omega(n^6)$.
Note that our quantitative bound asymptotically dominates the bounds of Abrahamson~\cref{prop:abrahamson=boolean} for all values of $S$; it always is at least as large (up to a constant factor) and
the only regimes where our quantitative bound does not strictly dominate that of Abrahamson are when $S$ is $\Theta(1)$ and when $S$ is $\Theta(n)$. 
Of course, Abrahamson's
lower bounds are for the branching program model which allows for the timing of each output bit to depend on the input.
(The classical lower bound of \cite{KSdW07} for output-oblivious query algorithms is exactly the same as that of Abrahamson for space $O(\sqrt{n})$.)
Abrahamson's bound on the number of queries becomes the trivial $\Theta(n^2)$ when $S=\Theta(n^{3/2})$ which is tight
for the distribution used in Abrahamson's paper, whereas the lower bound of~\cref{thm:classical-boolean-matrix} remains non-trivial so long as $S$ is $o(n^2)$.
In fact, just as with our quantum lower bound in~\cref{thm:quantum-booolean-matrix}, the exponents of $n$ and $S$ in \cref{thm:classical-boolean-matrix} are optimal for a circuit model that allows arbitrary gates between queries since that would allow
the circuit to simulate a decision tree of height $2n^2$
that reads and remembers the entire input and produces all of the outputs at its leaves; our lower bounds also apply to such a model.
See \cref{fig:boolean-mat-mult-comparisons} for a 
comparison of our lower bounds with those of prior work
for both classical and quantum computation.

\begin{figure}  
    \resizebox{0.49\textwidth}{!}{
        \begin{tikzpicture}
            \begin{axis}[xmin=-0.05, ymin=1, xmax=2.05,ymax=3.05, samples=5, xlabel={$\log_n S$}, ylabel={$\log_n T $}, title=Quantum Boolean Matrix Multiplication, legend pos=south west, legend style={draw=none}]
              \addplot[black, thick, domain=0:2] {2.5-x/4};
              \addlegendentry{\Cref{thm:quantum-booolean-matrix}}
              \addplot[black, thick, dotted, domain=0:1] {2.5-x/2};
              \addlegendentry{\cite{KSdW07}}
              \addplot[name path=lowerbound,domain=-10:10, draw=none] {2};
              \path[name path=axis] (axis cs:-10,1) -- (axis cs:5,1);
              \addplot [
                    thick,
                    color=gray,
                    fill=gray, 
                    fill opacity=0.2
                ]
                fill between[
                    of=lowerbound and axis,
                    soft clip={domain=-10:10},
                ];
                \node[circle,fill,inner sep=2pt] at (axis cs:0,2.5) {};
                \node[circle,fill,inner sep=2pt] at (axis cs:2,2) {};
            \end{axis}
        \end{tikzpicture}
        }
    \hspace{0.1cm}
    \resizebox{0.49\textwidth}{!}{
        \begin{tikzpicture}
            \begin{axis}[xmin=-0.05, ymin=1, xmax=2.05,ymax=3.05, samples=5, xlabel={$\log_n S$}, ylabel={$\log_n T$}, title=Classical Boolean Matrix Multiplication, legend pos=south west, legend style={draw=none}]
              \addplot[black, thick, domain=0:2] {3-x/2};
              \addlegendentry{\Cref{thm:classical-boolean-matrix}}
              \addplot[black, thick, dotted, domain=0:1] {3-x};
              \addlegendentry{\cite{KSdW07}}
              \addplot[black, thick, dashed, domain=1:3/2] {3.5-x};
              \addplot[black, thick, dashed, domain=1/2:1] {2.5};
              \addplot[black, thick, dashed, domain=0:1/2] {3-x};
              \node[circle,fill,inner sep=2pt] at (axis cs:0,3) {};
              \addplot[name path=lowerbound,domain=-10:5, draw=none] {2};
              \addlegendentry{\cite{DBLP:conf/focs/Abrahamson90}}
    
              \node[circle,fill,inner sep=2pt] at (axis cs:2,2) {};
              \path[name path=axis] (axis cs:-10,1) -- (axis cs:10,1);
              \addplot [
                    thick,
                    color=gray,
                    fill=gray, 
                    fill opacity=0.2
                ]
                fill between[
                    of=lowerbound and axis,
                    soft clip={domain=-10:10},
                ];
            \end{axis}
        \end{tikzpicture}
        }
    \caption{Comparison of our lower bounds for Boolean matrix multiplication with those of prior work for both quantum and classical computation.  The shaded region comes from the
    fact that the time must always be $\Omega(n^2)$.  
    The endpoints mark choices of parameters where the
    upper and lower bounds match.}
    \label{fig:boolean-mat-mult-comparisons}
\end{figure}

We can extend the above to get a matching lower bound on the classical cumulative memory complexity.
\begin{corollary}
    Any output-oblivious classical query algorithm computing $n\times n$ Boolean matrix-multiplication with $T$ queries and space $S$ with success probability more than $1/(2T)$ must have cumulative memory that is $\Omega(n^{6}/T)$.
\end{corollary}
\begin{proof}
    Using \cref{lem:general-classical-boolean} we can apply \cref{prop:TS-to-cm-lb} with $m'(n) = n^2, h(k,n)=(2 \epsilon/15) \sqrt{k}n$ and $K(n) = 2^{\gamma /2}$ to get that the cumulative memory must be
    \begin{equation*}
        \Omega(\min(n^6 / T, n^4)) = \Omega(n^6/T)
    \end{equation*}
    As $T$ must be $\Omega(n^2)$.
\end{proof}

Using the same proof idea as in \cref{cor:mat-square}, the bounds in \cref{thm:quantum-booolean-matrix,thm:classical-boolean-matrix} immediately imply lower bounds for Boolean matrix squaring.
\begin{corollary}\label{cor:bool-matrix-sq}
    Any quantum circuit computing $n \times n$ Boolean matrix squaring on all inputs with $T$ queries, space $S$, and success probability more than $2^{-S}$ must have $T$ that is $\Omega(n^{2.5}/S^{1/4})$. Any such output-oblivious classical query algorithm must have $T$ that is $\Omega(n^3 / S^{1/2})$.
    Quantum and classical circuits for Boolean matrix squaring with success probability larger than $1/(2T)$ must have cumulative memories $\Omega(n^{10}/T^3)$ or $\Omega(n^6/T)$ respectively.
\end{corollary}

\subsection{Boolean matrix-vector product}

Finally, we discuss the problem of quantum computation of Boolean matrix-vector product and the closely-associated problem of systems of linear inequalities. 
Here, rather than producing quantitative improvements which seem unlikely, we focus on a qualitative improvement in existing results.

Though \cite{DBLP:conf/focs/Abrahamson90} does not contain an explicit theorem statement on time-space tradeoffs for Boolean matrix-vector products that is the analog of the linear algebra bound in \cite{Abr91} or our \cref{thm:time-space-matrix-vector}, \cite{DBLP:conf/focs/Abrahamson90}
contains the claim that analogous results do indeed hold for this problem using the same ideas.   (The lower bound would be a factor $n$ smaller than the lower bound for linear algebra.)

For quantum circuits, Klauck, \v{S}palek, and de Wolf~\cite{KSdW07} prove the following results for computing Boolean matrix-vector products. (They also prove a similar result for the case of output-oblivious classical query algorithms, though that does not apply to unconstrained branching programs.)

\begin{proposition}[Theorem 23 in \cite{KSdW07}]\label{prop-bool-mat-vec}
For every $S$ in $o(n/\log n)$, there is an $n \times n$ Boolean matrix $A^{(S)}$ such that every bounded-error quantum circuit with space at most $S$ that computes Boolean matrix-vector product $A^{(S)} \bullet x$ in $T$ queries requires that $T$ is  $\Omega(\sqrt{n^{3}/S})=\Omega(n^{1.5}/S^{0.5})$.
\end{proposition}

This result is weaker than a standard time-space tradeoff since the function involved is not independent of the circuits that might compute
it.
In particular, \cite{KSdW07} does not find a single function that is hard for all space bounds, as the matrix $A^{(S)}$ that they use changes depending on the value of $S$. 
Because~\cite{KSdW07} does not express this dependence in the statement of their results, we provide a detailed discussion of their arguments to make the need for that dependence clear.
We will also need their definitions in  our results.

For $S=o(n/\log n)$, the matrix $A^{(S)}$ is produced
via the probabilistic method using the following distribution:
Choose $k$ to be a sufficiently large constant multiple of $S$.
This distribution chooses matrices $A\subseteq \{0,1\}^{n\times n}$ by
selecting a uniformly random subset of $n/(2k)$ positions in each row to set to 1, with the remainder of the entries in each row being 0.
They show that with positive probability over the choice of $A$,
for all sets $I\subseteq [n]$ of size $k$, at least $k/2$ of
the rows of $A_I$ contain at least $n/(6k)$ 1's that are unique in
their column of $A_I$; that is, those columns are
0 in all of the $k - 1$ other rows of $A_I$.
$A^{(S)}$ is then some fixed matrix for which this property is true.

More precisely, when we fix a row $j\in I$ and the $n/(2k)$ columns
where it is 1, the expected number of the $(k-1)n/(2k)<n/2$ 1's
among the rows in $I\setminus\{j\}$ that land in those $n/(2k)$ 
columns is less than $n/(4k)$.  By a Hoeffding bound, the number of
those 1's is at most $n/(3k)$ except with
probability exponentially small in $n/k$, which is $n^{-\omega(1)}$
since $k=O(S)=o(n/\log n)$.
Hence, except with probability $n^{-\omega(1)}$, a row $j\in I$ is \emph{good
for $I$} in that at least $n/(2k)-n/(3k)=n/(6k)$ of the 1's in row $j$ are unique in their respective columns in $A_I$.
For a fixed $I$, the probability that there is no $J\subseteq I$ of size $k/2$ all of whose rows are good for $I$ is less than the probability that there are $k/2$ rows of $I$ that are not
good for $I$.   This happens with probability at most $n^{-\omega(k)}$ since are at most $\binom{k}{k/2}$ such subsets of rows of size $k/2$, each of which is not good for $I$ with probability $n^{-\omega(k)}$ (and the probabilities are negatively associated).
Since there are only $\binom{n}{k}$ choices of $I$, the total
probability that $A$ does not have desired properties is only $n^{-\omega(k)}$.

The proof of \cref{prop-bool-mat-vec} follows from the usual time-space lower bound methodology and the following
lemma:

\begin{lemma}
There is an $\alpha>0$ such that for every quantum circuit $\mathcal{C}$ that makes at most $\alpha\sqrt{kn}$ queries to $x\in \{0,1\}^n$, the probability that $\mathcal{C}$ produces at least
$k$ correct output values of $A^{(S)}\bullet x$ is at most $2^{-\Omega(k)}$.
\end{lemma}

\begin{proof}
    Let $I\subseteq [n]$ be the set of indices of the first $k$ outputs of $A^{(S)}\bullet x$ produced by $\mathcal{C}$.
    Let $J\subseteq I$ be the set of size $k/2$ rows that are
    good for $I$ guaranteed by the properties of $A^{(S)}$.  
    We show that the probability that $\mathcal{C}$ produces
    all outputs even for the rows in $J$ is exponentially small in $k$:
    For each row $j\in J$ there is a set $C_j$ of $n/(6k)$ columns of
    $A^{(S)}_I$ where the unique 1 is in row $j$.
    Consider the restriction to input vectors $x\in \{0,1\}^n$ that are 0 outside of $\bigcup_{j\in J} C_j$. 
    Then the outputs for $j\in J$ are a direct product of $k/2$
    OR functions of size $n/(6k)$ on the bits of $\bigcup_{j\in J} C_j$.
    By a strong direct product theorem for OR (Theorem 14 of \cite{KSdW07}), for $\varepsilon$ a sufficiently small constant, any circuit of height at most 
    $\varepsilon (k/2) \sqrt{n/(6k)} =\varepsilon \sqrt{kn/24}$ is correct with probability at most
    $2^{-\gamma k}$ for some constant $\gamma>0$.
\end{proof}

On the algorithmic side, we have the following:

\begin{proposition}
    For every $c>0$ and every Boolean matrix $A\in \{0,1\}^{m\times n}$
    there is a quantum circuit using space $O(\log n)$ and
    time $O(mn^{1/2}\log m)$ that computes Boolean
    matrix-vector product $A\bullet x$ with error at most $m^{-c}$.
    More precisely, the algorithm runs in time
    $O(|A|_{1/2} \log m)$ where $|A|_{1/2}=\sum_{i=1}^m \sqrt{|A_i|_1}$.
\end{proposition}

\begin{proof}
    For each row in turn, run Grover's algorithm to 
    compute the OR of the bits indexed by the 1's of
    $A_i$, the $i$-th row of $A$ with probability of error at most $m^{-c-1}$ per row for a total error of at most $m^{-c}$.
\end{proof}

We note that for the fixed matrix $A^{(S)}$, each row has
$\Theta(n/S)$ 1's so $|A^{(S)}|_{1/2}=\Theta(n^{3/2}/S^{1/2})$.  
This is an odd situation in that the matrix $A^{(S)}$ designed to require large time for space $S$ algorithms can be solved in nearly the same time bound by space $O(\log n)$ algorithms.

\subparagraph{Systems of linear inequalities}
The same space-dependent matrix $A^{(S)}$ in \cref{prop-bool-mat-vec} was also used in \cite{ASdW09} for systems of inequalities.
\begin{proposition}[Theorem 11 in \cite{ASdW09}]\label{prop-system-of-inequality}
Let $\vec{b}$ be the length $n$ all-$b$ vector. For every $S$ in $\min(O(n/b), o(n/ \log n))$ there exists an $n \times n$ Boolean matrix $A^{(S)}$ such that every bounded error quantum circuit with space at most $S$ that decides the system $A^{(S)}x \geq \vec{b}$ of $n$ inequalities requires that $T$ is $\Omega(\sqrt{bn^3/S})$.
\end{proposition}
Similar to \cite{KSdW07} this matrix is used so that any quantum circuit that computes $A^{(S)}x \geq \vec{b}$ can be broken down into slices that solve independent instances of the $b$-threshold function.

\subsubsection*{Our results}

Using \cref{prop-bool-mat-vec}, we can obtain a time-space tradeoff lower bound for quantum computation of Boolean matrix-vector product that has an only slightly weaker lower bound in terms of the matrix dimensions but, unlike
the previous bound, defines a fixed computational problem whose definition
is independent of the space bound allowed.

\begin{theorem}
\label{cor:ts-boolean-matrix-vector}
    There is a fixed $m\times n$ Boolean matrix $A$ with $m\le n\log_2 n$
    such that for every $S$ that is $o(n/\log n)$
    every bounded-error quantum circuit with space at most $S$ that computes Boolean matrix-vector product $A \bullet x$ in $T$ queries requires that $T$ is $\Omega(\sqrt{n^{3}/S})$.
\end{theorem}

\begin{proof}
  The matrix $A$ consists of a stacked version of the matrices $A_{(S_i)}$ from \cref{prop-bool-mat-vec} for each choice of $S_i=2^i \log_2 n$ and $0\le i\le \log_2 n - 2\log_2 \log_2 n -\omega(1)$.
  Any quantum circuit computing $A\bullet x$ using space $S$ must compute $A^{(S_i)}\bullet x$ for some $S_i$ where $S_i\le S$ is within factor of 2 of
$S$.   It is easy to see that the construction of $A_{(S)}$ for
\cref{prop-bool-mat-vec} is flexible in terms
of the constant factor by which $k$ exceeds $S$ and hence computing
matrix $A^{(S_i)}\bullet x$ also requires time $T$ that
is $\Omega(\sqrt{n^{3}/S})$ as required.
\end{proof}

\subparagraph{Systems of linear inequalities}
This same matrix $A$ can be substituted into \cref{prop-system-of-inequality} to obtain a time-space tradeoff for systems of inequalities.
\begin{corollary}
Let $\vec{b}$ be the length $n$ all-$b$ vector.
There is a fixed $m \times n$ Boolean matrix $A$ with $m \leq n \log_2 n$ such that for every $S$ in $\min(O(n/b), o(n/ \log n))$ every bounded error quantum circuit with space at most $S$ that decides the system $Ax \geq \vec{b}$ requires $T$ that is $\Omega(\sqrt{bn^3/S})$.
\end{corollary}

\bibliographystyle{alpha}
\bibliography{sources}

\section{Deterministic query algorithms}\label{sec:query-algs}

Here we review the matching time-space space tradeoffs that match our quantum and classical lower bounds. Most of these results were mentioned in~\cite{Abr91} but are more fully sketched here.  In the following, for simplicity, we describe versions of several of these algorithms over finite fields
rather than finite subsets of size $d$ over arbitrary fields.
For the more general case, the output values are sums of
products of input values and may take more bits to represent; because of this the $\log p$ in our bounds below can be replaced by $O(\max{\log d, \log n})$.

The first gives classical algorithms for matrix-vector products matching~\cref{thm:time-space-matrix-vector}.

\begin{proposition}\label{prop:mat-vec-upper-bound-alg}
Let $A$ be any $n\times n$ matrix over a finite field $\F_p$.
For any $S \in [\log_2 n, n \log_2 p]$ there is a deterministic classical query algorithm computing the matrix vector product $f(x) = Ax$ for all inputs $x \in \F_p$ that uses space $S$ and only $O(n^2 \log p\ /S)$ queries to the input.
\end{proposition}

\begin{proof}
Let $s=S/\log_2 p$.  The query algorithm (which has the matrix $A$ encoded in it) reads one entry of the input $x$ at a time and maintains a block of $s$ different partial sums (using $s \log_2 p$ space). This algorithm produces $S$ outputs every $n$ queries and thus produces all outputs with $n^2/s=n^2 \log_2 p\ /S$ queries.
\end{proof}

Note that in the special case of computing the Discrete Fourier Transform (DFT)~(\cref{cor:DFT-ts-lb}), this deterministic query bound can be made explicit using standard operations:

\begin{proposition}[\cite{DBLP:journals/tit/SavageS78}]
    There is a deterministic classical algorithm computing
    the Discrete Fourier Transform (DFT) $DFT_n(x)=W x$ using
    space $S\ge \log_2 n$ and time $O(n^2/S + n\log S)$.
\end{proposition}

\begin{proof}
Assume without loss of generality that $S$ and $n$ are powers of 2 and we have $O(S)$ space. 
This follows by evaluating the graph of the fast Fourier transform (FFT) algorithm for computing the DFT as shown
in~\cref{fig:FFT}.
 \begin{figure}[hb]
\begin{center}
    \includegraphics[height=1.5in]{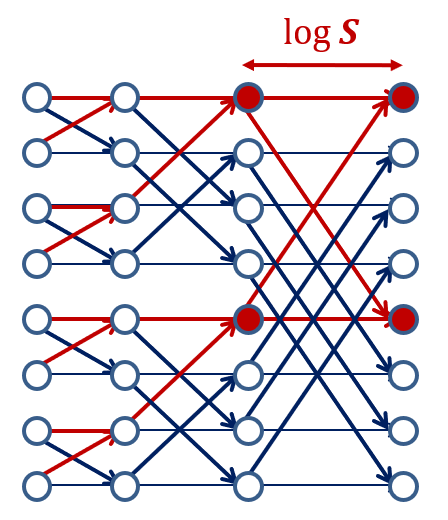}
    \caption{The FFT graph with the space-efficient evaluations on one pass
    highlighted.}\label{fig:FFT}
    \end{center}
\end{figure}
In a single pass over the input $x$ in $O(n+S\log S)$ time the algorithm can compute the values of $S$ of the outputs using space $O(S)$ as follows: 
while
maintaining $\log_2(n/S)\le S$ entries for the depth-first evaluation
of each subproblem at depth $\log_2 S$ and uses space $2S$
to iterate through the top $\log_2 S$ levels which are evaluated together in a size $S$ FFT computation.
This pass is repeated for each of the $n/S$ such blocks in turn.
\end{proof}

The following deterministic algorithms for convolution match~\cref{cor:convo-ts-lb}.

\begin{proposition}
For any $S\in [\log_2 n, n \log_2 p]$ there is a deterministic classical query algorithm that computes the convolution $f(u,v) = u * v$ where $u,v \in \F_p^{n}$ that uses $S$ space and only $O(n^2 \log p \ / S)$ queries.
\end{proposition}
\begin{proof}
    Let $s=S/ (2\log_2 p)$.
    The indices of $u,v$ and $w = u *v$ are reduced modulo $n$.
    The query algorithm computes outputs $w_i, \ldots w_{i + s}$ of the convolution as follows: 
    Initialize $w_i, \ldots w_{i+s}$ to the value zero.
    First query and record the values of $v_{i-1}, \ldots v_{i+s-1}$.
    Then query values of $u$ one at a time in increasing order ($u_1, u_2, \ldots u_n$).
    After reading $u_j$, for each $k \in \{i, \ldots, i+s\}$, add $u_{j} \cdot v_{k-j}$ to the value of $w_k$.
    Then forget the value of $v_{i+s-j}$ and query the value of $v_{i-j-1}$, remembering this value.
    After all of $u$ has been queried, we have that $w_k = \sum_{j\in[n]} u_jv_{k-j}$ which is the correct value for these outputs.
    Repeating this procedure $n/s$ times gives the convolution of $u$ and $v$ using only $S$ space and $2n$ queries per iteration.
    Since there are $n/s$ iterations, we get $O(n^2 \log p \ / S)$ queries.
\end{proof}

The algorithms below show that our matrix-inversion lower bound for upper-triangular matrices in~\cref{cor:mat-invert-ts-lb} cannot be improved for large space bounds, even for deterministic query algorithms.  This is open for small space bounds.

\begin{proposition}
    For any $S\in [n \log_2 p, n^2 \log_2 p]$ there is a deterministic classical query algorithm computing the inverse $f(A) = A^{-1}$ where $A\in \F_p^{n \times n}$ is a unit upper triangular matrix that uses $S$ space and only $O(n^4 \log p \ / S)$ queries.
\end{proposition}
\begin{proof}
    Let $s= S/(2 n \log p)$. We will produce columns $j_1, \ldots j_s$ of $A^{-1}$ as follows:
    Let $e_j$ be the column vector with entry $1$ at index $j$ and $0$ everywhere else.
    We use back substitution to solve the systems $Ax_1 = e_{j_1}, \ldots, Ax_s=e_{j_s}$ by querying each entry of $A$ exactly once. In particular, the $i$-th entry of $x_k$ is $1-\sum_{\ell \in [n-i]} A_{i,n-\ell+1} x_{n-\ell+1}$ when $i=k$ and $-\sum_{\ell \in [n-i]} A_{i,n-\ell+1} x_{n-\ell+1}$ otherwise.
    We start by computing the $n$-th entry of each $x_k$ and work backward toward the first entry.
    We record each entry of each $x_k$ as is it computed for use in the subsequent computational steps.
    Note that the $i$-th entries of all the $x_k$ only require making queries to the $i$-th row of $A$ and so all the $x_k$ can be computed with only $O(n^2)$ queries.
    Finally, each $x_k$ is output as the $j_k$-th column of $A^{-1}$.
    This procedure uses $O(n^2)$ queries and at most $S$ space to produce $s$ columns of the output.
    Thus the procedure must be repeated $n/s=2 n^2 \log p \ / S$ times to produce all $n$ columns of output.
    This gives a total query complexity of
    $O(n^4 \log p \ / S)$.
\end{proof}

The following give the deterministic algorithms matching our matrix-multiplication, Boolean matrix-multiplication  
(\cref{thm:mat-mult,thm:quantum-booolean-matrix}) and squaring lower bounds (\cref{cor:mat-square,cor:bool-matrix-sq}).

\begin{proposition}\label{prop:matrix-product-upperbound}
There are deterministic query algorithms for $n\times n$ Matrix Multiplication over $\mathbb{F}_p$ using space $S$ that make $O(n^3\sqrt{\log p}/\sqrt{S})$ queries.
Further, $O(n^3/\sqrt{S})$ queries suffice for deterministic algorithms using space $S$ to compute $n\times n$ Boolean Matrix Multiplication.
\end{proposition}

\begin{proof}
Let $s=S/(3\log p)$.   We partition
each input matrix $A$ and $B$ into 
$\sqrt{s}\times \sqrt{s}$ blocks $A_{ij}$ and $B_{ij}$
for $i,j\in [\ell]$ where $\ell=n/\sqrt{s}$.
We compute the $\sqrt{s}\times\sqrt{s}$ blocks $C_{ij}$ of
the product as follows:
Initialize the block $C_{ij}$ to $0$.
For $k=1$ to $\ell$, query all entries of $A_{ik}$ and $B_{kj}$ and add their product $A_{ik}B_{kj}$ to $C_{ij}$.
The 3 matrices $A_{ik}$, $B_{kj}$, and $C_{ij}$ together
require space $S$ since each entry can be expressed using
$\log p$ bits. The total number of queries to compute $C_{ij}$ is $n\sqrt{s}$ and there are $\ell^2=n^2/s$ blocks to
compute for a total of $n^3/\sqrt{s}=O(n^3 \sqrt{\log p}/\sqrt{S})$ queries as claimed.

The query algorithm for Boolean Matrix Multiplication is analogous
with $s=S/3$ and entry-wise $\lor$ instead of addition.
\end{proof}

Finally, we see that the matrix triple-product and cubing lower bounds in~\cref{cor:mat-mat-mat-product-ts-lb,cor:mat-cube-ts-lb} have matching deterministic query algorithms.

\begin{proposition}
For any $S\in [\log_2 n, n^2 \log_2 p]$ there is a deterministic classical query algorithm computing the Matrix Triple Product $f(A,B,C) = ABC$ where $A,B,C \in \F_p^{n \times n}$ that uses $S$ space and only $O(n^4 \log p \ / S)$ queries.
\end{proposition}
\begin{proof}
Let $s=S/(4 \log p)$.
We view the product $ABC$ as $(AB)C$ and use the same strategy as in \cref{prop:matrix-product-upperbound} to compute partial products of $(AB)$ and then $ABC$.
We partition the input, partial product, and output matrices into blocks $A_{ij}, B_{ij}, C_{ij}, (AB)_{ij},$ and $(ABC)_{ij}$ for $i,j \in [\ell]$ where $\ell = n/\sqrt{s}$.
To compute $(AB)_{ij}$ we initialize the values in the block to zero.
Then, for each $k \in [\ell]$, we query each $A_{ik}$ and $B_{kj}$ and then perform the multiplication of these submatrices, adding the result into $(AB)_{ij}$.
After iterating over all $k$, we have computed the value of $(AB)_{ij}$.
Now to compute $(ABC)_{ij}$ we start by initializing the values in $(ABC)_{ij}$ to zero.
For each $k \in [\ell]$, we first compute $(AB)_{ik}$ as a subroutine and then query $C_{kj}$ and add the partial product $(AB)_{ik}C_{kj}$ into $(ABC)_{ij}$.
After iterating over all $k$, we have computed the block $(ABC)_{ij}$.
This query algorithm stores at most $4$ different $\sqrt{s} \times \sqrt{s}$ blocks at any time step.
It requires $\sqrt{s} n$ queries to compute each $(AB)_{ij}$ and needs to compute $n/\sqrt{s}$ such blocks for each $(ABC)_{ij}$.
Adding the $\sqrt{s} $ queries to $C$ needed to compute $(ABC)_{ij}$ gives $n\sqrt{s} (1+n/\sqrt{s})$ total queries to compute each block $(ABC)_{ij}$.
Since there are $n^2 / s$ such blocks, we get $O(n^4 / s)$ or $O(n^4 \log p \ / S)$ queries.
\end{proof}

\end{document}